\RequirePackage[l2tabu,orthodox]{nag}
\documentclass
[11pt,letterpaper]
{article}

\usepackage{etex}
\usepackage{xspace,enumerate}
\usepackage[dvipsnames]{xcolor}
\usepackage[T1]{fontenc}
\usepackage[full]{textcomp}
\usepackage[american]{babel}
\usepackage{mathtools}
\usepackage{amsthm}
\newtheorem{theorem}{Theorem}[section]
\newtheorem*{theorem*}{Theorem}

\newtheorem*{proposition*}{Proposition}
\newtheorem{lemma}[theorem]{Lemma}
\newtheorem*{lemma*}{Lemma}
\newtheorem{corollary}[theorem]{Corollary}
\newtheorem*{conjecture*}{Conjecture}

\newtheorem*{fact*}{Fact}

\newtheorem*{hypothesis*}{Hypothesis}

\theoremstyle{definition}
\newtheorem{definition}[theorem]{Definition}
\newtheorem*{definition*}{Definition}

\theoremstyle{remark}
\newtheorem{claim}[theorem]{Claim}
\newtheorem*{claim*}{Claim}
\newtheorem{remark}[theorem]{Remark}
\newtheorem*{remark*}{Remark}

\newtheorem*{observation*}{Observation}
\usepackage[
letterpaper,
top=1.2in,
bottom=1.2in,
left=1in,
right=1in]{geometry}
\usepackage{newpxtext} 
\usepackage{textcomp} 
\usepackage[varg,bigdelims]{newpxmath}
\usepackage[scr=rsfso]{mathalfa}
\usepackage{bm} 
\let\mathbb\varmathbb
\usepackage{microtype}
\usepackage[
colorlinks=true,
urlcolor=blue,
linkcolor=blue,
citecolor=OliveGreen,
]{hyperref}
\usepackage[capitalise,nameinlink]{cleveref}
\crefname{lemma}{Lemma}{Lemmas}
\crefname{fact}{Fact}{Facts}
\crefname{theorem}{Theorem}{Theorems}
\crefname{corollary}{Corollary}{Corollaries}
\crefname{claim}{Claim}{Claims}
\crefname{example}{Example}{Examples}
\crefname{algorithm}{Algorithm}{Algorithms}
\crefname{problem}{Problem}{Problems}
\crefname{definition}{Definition}{Definitions}
\usepackage{paralist}
\usepackage{turnstile}
\usepackage{mdframed}
\usepackage{tikz}
\usepackage{caption}
\DeclareCaptionType{Algorithm}
\usepackage{newfloat}

\newcommand{\Authornotecolored}[3]{}
\newcommand{\Authorcomment}[2]{}
\newcommand{\Authorfnote}[2]{}

\definecolor{forestgreen(traditional)}{rgb}{0.0, 0.27, 0.13}

\usepackage{boxedminipage}

\newcommand{\set}[1]{\{#1\}}

\newcommand{\norm}[1]{\lVert#1\rVert}

\newcommand{\Psymb}{\mathbb{P}}

\DeclareMathOperator*{\ProbOp}{\Psymb}
\renewcommand{\Pr}{\ProbOp}

\newcommand\bdot\bullet

\newcommand{\N}{\mathbb N}
\newcommand{\R}{\mathbb R}

\newcommand{\cG}{\mathcal G}

\renewcommand{\leq}{\leqslant}
\renewcommand{\le}{\leqslant}
\renewcommand{\geq}{\geqslant}
\renewcommand{\ge}{\geqslant}
\let\epsilon=\varepsilon
\numberwithin{equation}{section}
\newcommand\MYcurrentlabel{xxx}
\newcommand{\MYstore}[2]{%
  \global\expandafter \def \csname MYMEMORY #1 \endcsname{#2}%
}
\newcommand{\MYload}[1]{%
  \csname MYMEMORY #1 \endcsname%
}
\newcommand{\MYnewlabel}[1]{%
  \renewcommand\MYcurrentlabel{#1}%
  \MYoldlabel{#1}%
}
\newcommand{\MYdummylabel}[1]{}
\newcommand{\torestate}[1]{%
  \let\MYoldlabel\label%
  \let\label\MYnewlabel%
  #1%
  \MYstore{\MYcurrentlabel}{#1}%
  \let\label\MYoldlabel%
}
\newcommand{\restatetheorem}[1]{%
  \let\MYoldlabel\label
  \let\label\MYdummylabel
  \begin{theorem*}[Restatement of \cref{#1}]
    \MYload{#1}
  \end{theorem*}
  \let\label\MYoldlabel
}
\newcommand{\restatelemma}[1]{%
  \let\MYoldlabel\label
  \let\label\MYdummylabel
  \begin{lemma*}[Restatement of \cref{#1}]
    \MYload{#1}
  \end{lemma*}
  \let\label\MYoldlabel
}
\newcommand{\restateprop}[1]{%
  \let\MYoldlabel\label
  \let\label\MYdummylabel
  \begin{proposition*}[Restatement of \cref{#1}]
    \MYload{#1}
  \end{proposition*}
  \let\label\MYoldlabel
}
\newcommand{\restatefact}[1]{%
  \let\MYoldlabel\label
  \let\label\MYdummylabel
  \begin{fact*}[Restatement of \prettyref{#1}]
    \MYload{#1}
  \end{fact*}
  \let\label\MYoldlabel
}
\newcommand{\restate}[1]{%
  \let\MYoldlabel\label
  \let\label\MYdummylabel
  \MYload{#1}
  \let\label\MYoldlabel
}

\newcommand{\eps}{\epsilon}

\allowdisplaybreaks
\newcommand*{\zo}{\set{0,1}}

\usepackage{amssymb,amsmath,amsthm,enumerate,nicefrac}
\usepackage{graphicx, color}
\usepackage{epstopdf}
\usepackage{bbm}

\AtBeginDocument{%
  \addtolength\abovedisplayskip{-0.25\baselineskip}%
  \addtolength\belowdisplayskip{-0.15\baselineskip}%
}

\usepackage{float}
\floatstyle{ruled}
\newfloat{algorithm}{tbp}{loa}
\providecommand{\algorithmname}{Algorithm}
\floatname{algorithm}{\protect\algorithmname}

\usepackage[titletoc,title]{appendix}
\usepackage{url}

\newtheorem{thmx}{Theorem}

\newtheorem{lemmax}{Lemma}

\newtheorem{definitionx}{Definition}

\def\XX{ X }
\def\AA{ A }

\def\N{{\mathbb {N}}}
\def\Reals{{\mathbb {R}}}

\newcommand{\Ex}{\mathop{\bf E\/}}

\def\R{{\mathcal {R}}}

\def\B{{\mathcal {B}}}

\def\L{{\mathcal {L}}}
\def\M{{\mathcal {M}}}

\def\b{\textbf{b}}

\def\P{{\mathbb {P}}}
\newcommand{\inner}[1]{\langle#1\rangle}

\def\Sk{[n]^{(k)}}

\usepackage{color}

\newcommand{\mm}{\tilde{m}}

\def\one{{\mathbbm {1}}}
\def\a{\bar{a}}
\def\pari{Xor}

\newcommand{\F}{\mathbb{F}}

\title{Time-Space Tradeoffs for Distinguishing Distributions and Applications to Security of Goldreich's PRG}
\author{Sumegha Garg\thanks{sumeghag@cs.princeton.edu. Department of Computer Science, Princeton University.}
\and
Pravesh K. Kothari%
\thanks{kotpravesh@gmail.com. Computer Science Department, Carnegie Mellon University. Part of this work was completed when Pravesh was at Princeton University
and Institute for Advanced Study, Princeton.
}
\and
Ran Raz%
\thanks{ran.raz.mail@gmail.com. Department of Computer Science, Princeton University. Research supported by the Simons Collaboration on Algorithms and Geometry, by a Simons Investigator Award and by the National Science Foundation grant No. CCF-1714779.
}
}

\date{}
\begin{document}
\maketitle
\begin{abstract}
In this work, we establish lower-bounds against memory bounded algorithms for distinguishing between natural pairs of related distributions from samples that arrive in a streaming setting. 

Our first result applies to the problem of distinguishing the uniform distribution on $\zo^n$ from uniform distribution on some unknown linear subspace of $\zo^n$. As a specific corollary, we show that any algorithm that distinguishes between uniform distribution on $\zo^n$ and uniform distribution on an $n/2$-dimensional linear subspace of $\{0,1\}^n$ with non-negligible advantage needs $2^{\Omega(n)}$ samples or $\Omega(n^2)$ memory (tight up to constants in the exponent). 

Our second result applies to distinguishing outputs of Goldreich's \emph{local} pseudorandom generator from the uniform distribution on the output domain. Specifically, Goldreich's pseudorandom generator $\cG$ fixes a predicate $P:\zo^k \rightarrow \zo$ and a collection of subsets $S_1, S_2, \ldots, S_m \subseteq [n]$ of size $k$. For any seed $x \in \zo^n$, it outputs $P(x_{S_1}), P(x_{S_2}), \ldots, P(x_{S_m})$ where $x_{S_i}$ is the projection of $x$ to the coordinates in $S_i$. We prove that whenever $P$ is $t$-resilient (all non-zero Fourier coefficients of $(-1)^P$ are of degree $t$ or higher), then no algorithm, with $<n^\epsilon$ memory, can distinguish the output of $\cG$ from the uniform distribution on $\zo^m$ with a large inverse polynomial advantage, for stretch $m \le \left(\frac{n}{t}\right)^{\frac{(1-\epsilon)}{36}\cdot t}$ (barring some restrictions on $k$). The lower bound holds in the streaming model where at each time step $i$, $S_i\subseteq [n]$ is a randomly chosen (ordered) subset of size $k$ and the distinguisher sees either $P(x_{S_i})$ or a uniformly random bit along with $S_i$. 

An important implication of our second result is the security of Goldreich's generator with super linear stretch (\emph{in the streaming model}), against memory-bounded adversaries, whenever the predicate $P$ satisfies the \emph{necessary} condition of t-resiliency identified in various prior works. 

Our proof builds on the recently developed machinery for proving time-space trade-offs (Raz 2016 and follow-ups). Our key technical contribution is to adapt this machinery to work for \emph{distinguishing} problems in contrast to prior works on similar results for \emph{search}/learning problems.

\end{abstract}
\clearpage

\section{Introduction}
\label{sec:intro}

This work is motivated by the following basic question: suppose an algorithm is provided with a stream of $m$ i.i.d. samples from a random source. What's the minimum memory required to decide whether the source is ``truly random''  or ``pseudorandom''?

Algorithmically distinguishing perfect randomness from pseudorandomness naturally arises in the context of learning theory (and can even be equivalent to learning in certain models~\cite{DBLP:conf/stoc/Daniely16,DBLP:conf/colt/DanielyS16,DBLP:conf/colt/Vadhan17,DBLP:conf/innovations/KothariL18}), pseudorandomness and cryptography.

There has been a surge of progress in proving lower bounds for memory-bounded streaming algorithms beginning with Shamir~\cite{DBLP:conf/nips/Shamir14} and Steinhardt-Valiant-Wager~\cite{DBLP:conf/colt/SteinhardtVW16} who conjectured a $\Omega(n^2)$ memory lower bound for learning parity functions with $2^{o(n)}$ samples. This conjecture was proven in~\cite{Raz16}. In a follow up work, this was generalized to learning sparse parities in \cite{KRT17} and more general learning problems in~\cite{Raz17,GRT18,DBLP:conf/colt/MoshkovitzM17,BGY18,DS,MT,MM2,SSV19,GRT19}. 

All of these lower bounds hold for learning (more generally, \emph{search}) problems that ask to identify an unknown member of a target function class from samples. In this work, we build on the progress above and develop techniques to show lower bounds for apparently easier task of simply distinguishing uniformly distributed samples from pseudorandom ones. \cite{dgkr19} studies the related problem of distribution testing under communication and memory constraints. \cite{dgkr19} gave a one-pass streaming algorithm (and a matching lower bound for a broad range of parameters) for uniformity testing on $[N]$ that uses $m$ memory and $O(N\log(N)/(m\eps^4))$ samples for distinguishing between uniform distribution on $[N]$ and any distribution that is $\eps$-far from uniform. 

As we next discuss, our results have consequences of interest in cryptography (ruling out  memory-bounded attacks on Goldreich's pseudorandom generator~\cite{Goldreich00} in the streaming model) and average-case complexity (unconditional lower bounds on the number of samples needed, for memory-bounded algorithms, to refute random constraint satisfaction problems, in the streaming model).

\subsection{Our Results}
We now describe our results in more detail. Our main results show memory-sample trade-offs for distinguishing between truly random and pseudorandom sources for the following two settings: 

\begin{enumerate}
\item \textbf{Uniform vs $k$-Subspace Source}: The pseudorandom subspace source of dimension $k$ chooses some arbitrary $k$-dimensional linear subspace $S \subseteq \zo^n$ and draws points uniformly from $S$. The truly random source draws points uniformly from $\zo^n$. 

\item \textbf{Uniform vs Local Pseudorandom Source}: The pseudorandom source fixes a $k$-ary Boolean predicate $P:\zo^k \rightarrow \zo$. It chooses a uniformly random $x \in \zo^n$ and generates samples $(\alpha,b) \in [n]^{(k)} \times \zo$ where $[n]^{(k)}$ represents the set of all ordered $k$-tuples with exactly $k$ elements from $[n]$ and $\alpha$ is chosen uniformly at random from $[n]^{(k)}$ and $b$ is the evaluation of $P$ at $x^{
\alpha}$ - the $k$-bit string obtained by projecting $x$ onto the coordinates indicated by $\alpha$. The truly random source generates samples $(\alpha,b)$ where $\alpha \in [n]^{(k)}$ and $b \in \zo$ are chosen uniformly and independently. 
\end{enumerate} 

We model our algorithm by a \emph{read-once branching program} (ROBP) of width $2^b$ (or memory $b$) and length $m$. Such a model captures any algorithm that takes as input a stream of $m$ samples and has a memory of at most $b$ bits. Observe that there's no  restriction on the computation done at any node of an ROBP.

Roughly speaking, this model gives the algorithm unbounded computational power and bounds only its memory-size and the number of samples used.

Our first main result shows a lower bound on memory-bounded ROBPs for distinguishing between uniform and $k$-subspace sources.
\begin{theorem}[Uniform vs Subspace Sources]\label{thm:introsubspace}
Any algorithm that distinguishes between uniform and subspace source of dimension $k$ (assuming $k>c\log n$ for some large enough constant $c$) with probability at least $1/2 +2^{-o(k)}$ requires either a memory of $\Omega(k^2)$ or at least $2^{\Omega(k)}$ samples. In particular, distinguishing between the uniform distribution on $\zo^n$ and the uniform distribution on an unkown linear subspace of dimension $n/2$ in $\zo^n$ requires $\Omega(n^2)$ memory or $2^{\Omega(n)}$ samples. 
\end{theorem}

Crouch et. al.~\cite{DBLP:conf/esa/CrouchMVW16} recently proved that any algorithm that uses at most $n/16$ bits of space requires $\Omega(2^{n/16})$ samples to distinguish between uniform source and a subspace source of dimension $k = n/2$. 
They suggest the question of improving the space bound to $\Omega(n^2)$ while noting that their techniques do not suffice. 
For $k = \Theta(n)$, our lower bound shows that any algorithm with memory at most $cn^2$ for some absolute constant $c$ requires $2^{\Omega(n)}$ samples. 
This resolves their question. 

\textbf{Upper bound}: In Section~\ref{sec:reduction}, we exhibit a simple explicit branching program that uses $2^{O(k)}$ samples and $O(1)$ memory to succeed in solving the distinguishing problem with probability $3/4$. We also show a simple algorithm that uses $O(k^2)$ memory and $O(k)$ samples, and succeeds in solving the distinguishing problem with probability $3/4$. \emph{Thus, in the branching program model, the lower bound is tight up to constants in the exponent}. \\

Our second main result gives a memory-sample trade-off for the uniform vs local pseudorandom source problem for all predicates that have a certain well-studied pseudorandom property studied in cryptography under the name of \emph{resilience}.

A $k$-ary Boolean function $P$ is said to be $t$\emph{-resilient} if $t$ is the maximum integer such that $(-1)^P$ (taking the range of the boolean function to be \{-1,1\}) has zero correlation with every parity function of at most $t-1$ out of $k$ bits. In particular, the parity function on $k$ bits is $k$-resilient.

\begin{theorem}[Uniform vs Local Pseudorandom Sources] \label{thm:local-arbit-predicate}
Let $0<\epsilon<1-3\frac{\log 24}{\log n}$ and $P$ be a $t$-resilient $k$-ary predicate for $k < n^{(1-\epsilon)/6}/3, n/c$\footnote{$c$ is a large enough constant}. Then, any ROBP that succeeds with probability at least $1/2 +\Omega(\left(\frac{t}{n}\right)^{\Omega(t \cdot (1-\epsilon))})$ at distinguishing between uniform and local pseudorandom source for predicate $P$, requires $\left(\frac{n}{t}\right)^{\Omega(t \cdot (1-\epsilon))}$ samples or $n^\eps$ memory. 
\end{theorem}

\textbf{Upper bound}: In Subsection \ref{section:tightness3}, we give an algorithm that takes $(n^{\eps}+k)k\log n$ memory and $(n^{(1-\eps)k})(n^\eps+k)$ samples, and distinguishes between uniform and local pseudorandom source for any predicate $P$, with probability $99/100$. Thus, the lower bounds are almost tight up to $\log n$ factors and constant factors in the exponent for certain predicates ($t=\Omega(k)$). The question of whether there exists a better algorithm that runs in $O(n^{(1-\eps)t})$ samples and $O(n^\eps)$ memory, and distinguishes between uniform and local pseudorandom source with high probability, for $t$-resilient predicates $P$, remains open.\\

This result has interesting implications for well-studied algorithmic questions in average-case complexity and cryptography such as refuting random constraint satisfaction~\cite{MR2121179-Feige02,MR3473335-Allen15,DBLP:conf/stoc/RaghavendraRS17,DBLP:conf/stoc/KothariMOW17}) and existence of local pseudorandom generators~\cite{CryanM01,MR2238029-Mossel06,DBLP:conf/eurocrypt/BarakBKK18,LombardiV17,Applebaum13,Applebaum16} with super linear stretch where a significant effort has focused on proving lower bounds on various restricted models such as propositional and algebraic proof systems, spectral methods, algebraic methods and semidefinite programming hierarchies. While bounded memory attacks are well-explored in cryptography~\cite{DBLP:journals/joc/Maurer92,DBLP:conf/crypto/CachinM97,DBLP:conf/crypto/AumannR99,DBLP:journals/tit/AumannDR02,DBLP:conf/crypto/Vadhan03,DBLP:conf/eurocrypt/DziembowskiM04,Raz16,VV}, to the best of our knowledge, memory has not been studied as explicit resource in this context. We discuss these applications further in the paper.

For the special case when $P(x) =\sum_{i}^k x^i \mod 2$, the parity function on $k$ bits, we can prove stronger results for a wider range of parameters. 

\begin{theorem}[Uniform vs Local Pseudorandom Sources with Parity Predicate]
Let $0<\epsilon<1-3\frac{\log 24}{\log n}$ and $P$ be the parity predicate on $k$ bits for $0<k<n/c$ ($c$ is a large enough constant). Suppose there's a ROBP that distinguishes between uniform and local pseudorandom source for the parity predicate, with probability at least $1/2 +s$ and uses $<n^{\epsilon}$ memory. If $s > \Omega\left(\left(\frac{k}{n}\right)^{\Omega((1-\eps)\cdot k)}\right)$, then, the ROBP requires $\left(\frac{n}{k}\right)^{(\Omega((1-\eps)\cdot k)}$ samples. 
\label{thm:local-parity}

\end{theorem}

The above results show lower bounds for \emph{sublinear} memory algorithms. 
For a slight variant of the above uniform vs local pseudorandom source problem, we can in fact upgrade our results to obtain the following lower bounds against \emph{super-linear} memory algorithms. See Section~\ref{sec:reduction} for details.
\begin{theorem}\label{thm:introsuperlinearparity}
For large enough $n$ and $k>c\log n$ (where $c$ is a large enough constant) and $k\le \frac{n}{4}$, any algorithm that can distinguish satisfiable sparse parities of sparsity $k$ on $n$ variables (of type 
$(a,b)=(a^1,a^2,...,a^n,b)\in\{0,1\}^{n+1}$, where $\forall i\in[n],\; a^i=1$ with probability 
$\frac{k}{n}$ and $b=\inner{a,x}$) from random ones (of similar type $(a,b)$ but $b$ is now chosen uniformly at random from $\{0,1\}$),
with success probability at least $\frac{1}{2}+2^{-o(k)}$, requires
either a memory of size $\Omega(nk)$ or $2^{\Omega(k)}$ samples. 
\end{theorem}
In Remark~\ref{rem:tightness-2}, we observe that the above theorem is almost tight. Specifically, we observe that there are ROBPs that use a constant memory and $O(n 2^{O(k)})$ samples or $O(nk log n)$ memory and $O(n)$ samples to distinguish uniform sources from locally pseudorandom ones. 

\subsection{Applications to Security of Goldreich's Pseudorandom Generator}

A fundamental goal in cryptography is to produce secure constructions of cryptographic primities that are highly efficient. In line with this goal, Goldreich~\cite{DBLP:journals/iacr/Goldreich00a} proposed a candidate one-way function given by the following pseudorandom mapping that takes $n$-bit input $x$ and outputs $m$ bits: fix a \emph{predicate} $P:\zo^k \rightarrow \zo$, pick $a_1, a_2, \ldots, a_m$ uniformly at random\footnote{More generally, Goldreich proposed that $a_1, a_2, \ldots, a_m$ could be chosen in a pseudorandom way so as to ensure a certain ``expansion'' property. We omit a detailed discussion here.} from $[n]^{(k)}$ and output $P(x^{a_1}), P(x^{a_2}),\ldots,P(x^{a_m})$. Here, $a_1, \ldots, a_m$ and $P$ are public and the seed $x$ is secret. Later works (starting with \cite{DBLP:conf/focs/MosselST03}) suggested using this candidate as pseudorandom generator. 

The main question of interest is the precise trade-off between the \emph{locality} $k$ and the \emph{stretch} $m$ for a suitable choice of the predicate $P$. In several applications, we need that the generator has a super-linear stretch (i.e. $m = n^{1+ \delta}$ for some $\delta >0$) with constant locality (i.e. $k = O(1)$). 

The simplicity and efficiency of such a candidate is of obvious appeal. This simplicity has been exploited to yield a host of applications including public-key cryptography from combinatorial assumptions~\cite{MR2743266-Applebaum10}, highly efficient secure multiparty computation~\cite{DBLP:conf/eurocrypt/IshaiKOPS11} and most recently, basing \emph{indistinguishability obfuscation} on milder assumptions~\cite{DBLP:conf/eurocrypt/Lin16,cryptoeprint:2015:730,DBLP:conf/focs/LinV16,cryptoeprint:2016:1096,DBLP:conf/crypto/LinT17}. 

Evidence for the security of Goldreich's candidate has been based on analyzing natural classes of attacks based on propositional proof systems~\cite{MR2114305-Alekhnovich04}, spectral methods and semidefinite programming hierarchies~\cite{MR3280991-ODonnell14,MR3473335-Allen15,MR3388187-Barak15,DBLP:conf/stoc/KothariMOW17,LombardiV17a,DBLP:conf/eurocrypt/BarakBKK18} and algebraic methods~\cite{DBLP:journals/joc/ApplebaumBR16,DBLP:conf/stoc/ApplebaumL16}. In particular, previous works~\cite{DBLP:conf/stoc/KothariMOW17,DBLP:conf/stoc/ApplebaumL16} identified $t$-resiliency of the predicate $P$ as a necessary condition for the security of the candidate for $m= n^{\Omega(t)}$ stretch. 

The \emph{uniform vs local pseudorandom source} problem considered in this work is easily seen as the algorithmic question of distinguishing the output stream generated by Goldreich's candidate generator from a uniformly random sequence of bits. In particular, our results imply security of Goldreich's candidate against bounded memory algorithms for super-linear stretch when instantiated with any $t$-resilient predicate for large enough constant $t$ (but in the streaming model). Goldreich's candidate generator would fix the sets $a_1, a_2, \ldots, a_m$ (which are public) and output $P(x^{a_1}), P(x^{a_2}),\ldots,P(x^{a_m})$ for $n$ sized input $x$ ($m>n$). We prove the security of Goldreich's generator in the model where $a_1, a_2, \ldots, a_m$, still public, are chosen uniformly at random from $[n]^{(k)}$ and streamed with the generated bits.

We note that our lower bounds continue to hold even when the locality $k$ grows polynomially with the seed length $n$. 

\begin{corollary}[Corollary of Theorem~\ref{thm:local-arbit-predicate}]
Let $0<\epsilon<1-3\frac{\log 24}{\log n}$ and $P$ be a $t$-resilient $k$-ary predicate for $k = O(n^{(1-\epsilon)/6})$. Then, Goldreich's PRG, when instantiated with any $t$-resilient $k$-ary predicate $P$ such that $k\ge  t > 36$ and stretch $m =(n/t)^{O(t)(1-\epsilon)}$, is secure against all read-once branching programs with memory-size bounded from above by $n^{\epsilon}$, in the streaming model.
\end{corollary}

\subsection{Applications to Refuting Random CSPs }
Theorems~\ref{thm:local-arbit-predicate} and ~\ref{thm:local-parity} can also be intepreted as lower bounds for the problem of refuting random constraint satisfaction problems. 

A random CSP with predicate $P:\zo^k \rightarrow \zo$ is a constraint satisfaction problem on $n$ variables $x \in \zo^n$. 
More relevant to us is the variant where the constraints are randomly generated as follows: choose an ordered $k$-tuple of variables $a$ from $[n]$ at random, a bit $b \in \zo$ at random and impose a constraint $P(x^a) = b$. When the number of constraints $m \gg n$, the resulting instance is unsatisfiable with high probability for any non-constant predicate $P$. The natural algorithmic problem in this regime is that of \emph{refutation} - efficiently finding a short witness that certifies that the given instance is far from satisfiable. It is then easy to note that the uniform vs local pseudorandom source problem is the task of distinguishing between constraints in a random CSP (with predicate $P$) and one with a satisfying assignment. Note that refutation is formally harder than the task of distinguishing between a random CSP and one that has a satisfying assignment.

Starting with investigating in proof complexity, random CSPs have been intensively studied in the past three decades. When $P$ is $t$-resilient for $t \geq 3$, all known efficient algorithms~\cite{MR3473335-Allen15} require $m \gg n^{1.5}$ samples for the refutation problem. This issue was brought to the forefront in ~\cite{MR2121179-Feige02} where Feige made the famous ``Feige's Hypothesis'' conjecturing the impossibility of refuting random 3SAT in polynomial time with $\Theta(n)$ samples. Variants of Feige's hypothesis for other predicates have been used to derive hardness results in both supervised~\cite{DBLP:conf/nips/DanielyLS13,DBLP:conf/colt/DanielyLS14,DBLP:conf/stoc/DanielyLS14} and unsupervised machine learning~\cite{DBLP:conf/colt/BarakM16}.

In ~\cite{MR3280991-ODonnell14}, $t$-resilience was noted as a necessary condition for the refutation problem to be hard. Our Theorems~\ref{thm:local-arbit-predicate} and ~\ref{thm:local-parity} confirm this as a sufficient condition for showing lower-bounds for the refutation (in fact, even for the easier ``distinguishing'' variant) of random CSPs, with $t$-resilient predicates, in the streaming model with bounded memory.

\section{Preliminaries}
\label{sec:prelim}
Denote by log the logarithm to base 2. We use $Ber(p)$ to denote the Bernoulli distribution with parameter $p$ (probability of being 1). We use
$[n]$ to denote the set $\{1,2,...,n\}$.

For a random variable $Z$ and an event $E$,
we denote by $\P_Z$ the distribution of the random variables $Z$, and
we denote by $\P_{Z|E}$ the distribution of the random variable $Z$ conditioned on the event $E$.

Given an $n-$bit vector $y\in\{0,1\}^n$, we use $y^{i}$ to denote the $i^{th}$ coordinate of $y$, that is, $y=(y^1,y^2,...,y^n)$. We use $y^{-i}\in\{0,1\}^{n-1}$ to denote $y$ but with the $i$th coordinate deleted. Given two $n-$bit vectors $y,y'$, we use $\inner{y,y'}$ to denote the inner product of $y$ and $y'$ modulo 2, that is, $\inner{y,y'}=\sum_{i=1}^n y^{i}y'^{i} \mod 2$. We use $|y|$ to denote the number of ones in the vector $y$.

Given a set $S$, we use $y\in_R S$ to denote the random process of picking $y$ uniformly at random from the set $S$. Given a probability distribution $D$, we use $y\sim D$ to denote the random process of sampling $y$ according to the distribution $D$.

Next, we restate (for convenience) the definitions and results from previous papers \cite{Raz16,Raz17,KRT17,GRT18} that we use. 
\subsubsection*{Viewing a Learning Problem as a Matrix}
Let $\XX$, $\AA$ be two finite sets of size larger than 1.


Let $M: \AA \times \XX \rightarrow \{-1,1\}$ be a matrix.
The matrix $M$ corresponds to the following learning problem:
There is an unknown element $x \in \XX$ that was chosen uniformly at random. A learner tries to learn $x$ from samples
$(a, b)$, where $a \in \AA$ is chosen uniformly at random and $b = M(a,x)$.
That is, the learning algorithm is given a stream of samples,
$(a_1, b_1), (a_2, b_2) \ldots$, where each~$a_t$ is uniformly distributed and for every $t$, $b_t = M(a_t,x)$.

These papers model the learner for the learning problem corresponding to the matrix $M$ using a branching program:
\begin{definitionx} {\bf Branching Program for a Learning Problem:}
A branching program of length $m$ and width $d$, for learning, is a directed (multi) graph with vertices arranged in $m+1$ layers containing at most $d$ vertices each. In the first layer, that we think of as layer~0, there is only one vertex, called the start vertex.
A vertex of outdegree~0 is called a  leaf.
All vertices in the last layer are leaves
(but there may be additional leaves).
Every non-leaf vertex in the program has $2|\AA|$ outgoing edges, labeled by elements
$(a,b) \in \AA \times \{-1,1\}$, with exactly one edge labeled by each such $(a,b)$, and all these edges going
into vertices in the next layer.
Each leaf $v$ in the program is labeled by an element $\tilde{x}(v) \in \XX$, that
we think of as the output of the program on that leaf.

{\bf Computation-Path:} The samples
$(a_1, b_1), \ldots, (a_m, b_m) \in \AA \times \{-1,1\}$
that are given as input,
define a
computation-path in the branching
program, by starting from the start vertex
and following at
step~$t$ the edge labeled by~$(a_t, b_t)$, until reaching a leaf.
The program outputs the label $\tilde{x}(v)$ of the leaf $v$ reached by the computation-path.

{\bf Success Probability:}
The success probability of the program is the probability that $\tilde{x}=x$,
where $\tilde{x}$ is the element that the program outputs, and the probability is over $x,a_1,\ldots,a_m$ (where $x$ is uniformly distributed over $\XX$ and $a_1,\ldots,a_m$ are uniformly distributed over $\AA$, and for every $t$, $b_t = M(a_t,x)$).

\end{definitionx}

\begin{thmx}\label{thm:parity} \cite{Raz16,Raz17,GRT18} 
Any branching program that learns $x \in \{0,1\}^n$, from random linear equations over $\F_2$ 
with success probability $2^{-cn}$ requires either a width
of $2^{\Omega(n^2)}$ or a length of $2^{\Omega(n)}$ (where $c$ is a small enough constant).

\end{thmx}

\begin{thmx}\label{thm:sparse} \cite{GRT18} 
Any branching program that learns $x\in \{0,1\}^n$, from random sparse linear equations, 
of sparsity exactly $\ell$, over $\F_2$ with success probability  $2^{-cl}$ (where $c$ is a small enough constant) requires:
\begin{enumerate}
	\item Assuming $\ell \le n/2$: either a width
of size $2^{\Omega(n \cdot \ell)}$ or length of $2^{\Omega(\ell)}$.
	\item Assuming $\ell \le n^{0.9}$:  either a width
of size  $\Omega(n \cdot \ell^{0.99})$ or length of $\ell^{\Omega(\ell)}$.
\end{enumerate}
\end{thmx}

\subsubsection*{Norms and Inner Products}
Let $p \geq 1$.
For a function
$f: \XX \rightarrow \Reals$,
denote by $\norm{f}_{p}$ the $\ell_p$ norm of $f$, with respect to the  uniform distribution over $\XX$, that is:
$$\norm{f}_{p} =
\left( \Ex_{x \in_R \XX} \left[ |f(x)|^{p} \right] \right)^{1/p}.$$

For two functions
$f,g: \XX \rightarrow \Reals$, define their inner product with respect to the uniform distribution over $X$ as
$$\langle f,g \rangle =
\Ex_{x \in_R \XX} [ f(x) \cdot g(x) ].$$

For a matrix $M: \AA \times \XX \to \Reals$ and a row $a \in \AA$, we denote by $M_a: \XX \to \Reals$ the function corresponding to the $a$-th row of $M$. Note that for a function $f: \XX \to \Reals$, we have $\inner{M_a, f} = \frac{(M \cdot f)_a}{|X|}$.

\begin{definitionx} \cite{GRT18} {\bf 	$L_2$-Extractor:} \label{definition:l2-extractor} 
Let $\XX,\AA$ be two finite sets.
A matrix $M: \AA \times \XX \to \{-1,1\}$ is a $(k',\ell')$-$L_2$-Extractor with error $2^{-r'}$, if for every non-negative $f : \XX \to \Reals$ with $\frac{\norm{f}_2}{\norm{f}_1} \le 2^{\ell'}$ there are at most $2^{-k'} \cdot |A|$ rows $a$ in $A$ with
\[
\frac{|\inner{M_a,f}|}{\norm{f}_1}
\ge 2^{-r'}\;.
\]
\end{definitionx}

\begin{lemmax}\cite{KRT17}\label{lem:spafour}
Let $T_l$ be the set of $n$-bit vectors with sparsity exactly-$l$ for $l\in\mathbb{N}$, that is, $T_l=\{x\in\{0,1\}^n\mid \sum_{i=1}^nx^i=l\}$. Let $\delta\in (0,1]$. Let $\B_{T_l}(\delta)=\{\alpha\in\{0,1\}^n\mid \left|\Ex_{x\in T_l}(-1)^{\inner{\alpha,x}}\right|> \delta\}$. Then, for $\delta\ge (\frac{8l}{n})^{\frac{l}{2}}$,
\[|\B_{T_l}(\delta)|\le 2e^{-\delta^{2/l}\cdot n/8}\cdot 2^n\]
\end{lemmax}

\subsubsection*{Branching Program for a Distinguishing Problem}
Let $\XX$, $\AA$ be two finite sets of size larger than 1. Let $D_0$ be a distribution over the sample space $|\AA|$. 
Let $\{D_1(x)\}_{x\in\XX}$ be a set of distributions over the sample space $|\AA|$. 
Consider the following distinguishing problem: An unknown $\b\in\{0,1\}$ is chosen uniformly at random. 
If $\b=0$, the distinguisher is given independent samples from $D_0$. If $\b=1$, an unknown $x\in\XX$ is chosen uniformly at random, and
the distinguisher is given independent samples from $D_1(x)$. The distinguisher tries to learn $\b$ from the samples drawn according 
to the respective distributions. 

Formally, we model the distinguisher by a branching program as follows.

\begin{definition} {\bf Branching Program for a Distinguishing Problem:}
A branching program of length $m$ and width $d$, for distinguishing, is a directed (multi) graph with vertices 
arranged in $m+1$ layers containing at most $d$ vertices each. In the first layer, that we think of as layer~0, 
there is only one vertex, called the start vertex.
A vertex of outdegree~0 is called a  leaf.
All vertices in the last layer are leaves
(but there may be additional leaves).
Every non-leaf vertex in the program has $|\AA|$ outgoing edges, labeled by elements
$a\in\AA$, with exactly one edge labeled by each such $a$, and all these edges going
into vertices in the next layer.
Each leaf $v$ in the program is labeled by a $\tilde{\b}(v)\in\{0,1\}$, that
we think of as the output of the program on that leaf.

{\bf Computation-Path:} The samples
$a_1, \ldots, a_m \in \AA$
that are given as input,
define a
computation-path in the branching
program, by starting from the start vertex
and following at
step~$t$ the edge labeled by~$a_t$, until reaching a leaf.
The program outputs the label $\tilde{\b}(v)$ of the leaf $v$ reached by the computation-path.

{\bf Success Probability:}
The success probability of the program is the probability that $\tilde{\b}=\b$,
where $\tilde{\b}$ is the element that the program outputs, and the probability is over $\b,x,a_1,\ldots,a_m$ (where $\b$ is uniformly distributed over $\{0,1\}$, $x$ is uniformly distributed over $\XX$ and $a_1,\ldots,a_m$ are independently drawn from $D_0$ if $\b=0$ and $D_1(x)$ if $\b=1$).

\end{definition}

\section{Overview of the Proofs}
\label{sec:overview}
We prove our theorems using two different techniques. We prove Theorems \ref{thm:introsubspace} and \ref{thm:introsuperlinearparity} through reductions to the memory-sample lower bounds for the corresponding learning problems in Section \ref{sec:reduction}. Informally, for Theorem  \ref{thm:introsubspace}, we construct a branching program that learns the unknown vector $x$ from random linear equations in $\F_2$ by guessing each bit one by one sequentially and using the distinguisher, for distinguishing subspaces from uniform, to check if it guessed correctly. Then, we are able to lift the previously-known memory-sample lower bounds for the learning problem (Theorem \ref{thm:parity}) to the distinguishing problem. Similarly, we lift the memory-sample lower bounds for a variant of the learning problem in Theorem \ref{thm:sparse} to the get Theorem \ref{thm:introsuperlinearparity}.\\

We prove Theorems \ref{thm:local-arbit-predicate} and \ref{thm:local-parity} in Section \ref{sec:sublinear}.  Recall, a pseudorandom source fixes a $k$-ary Boolean predicate $P:\zo^k \rightarrow \zo$. It chooses a uniformly random $x \in \zo^n$ and generates samples $(\alpha,b) \in [n]^{(k)} \times \{0,1\}$ where $\alpha$ is a uniformly random (ordered) $k$-tuple of indices in $[n]$ and $b$ is the evaluation of $P$ at $x^{
\alpha}$ - the $k$-bit string obtained by projecting $x$ onto the coordinates indicated by $\alpha$. The truly random source samples $(\alpha,b)$ where $\alpha \in [n]^{(k)}$ and $b \in \zo$ are chosen uniformly and independently. The problem for a distinguisher is to correctly guess whether the $m$ samples are generated by a pseudorandom or a uniform source, when the samples arrive in a stream. We first show through a hybrid argument that a distinguisher $A$ that distinguishes between the uniform and pseudorandom source, with an advantage of $s$ over $1/2$, can also distinguish (with advantage of at least $s/m$) when only the $j$th (for some $j$) sample is drawn from the ``unknown source", the first $j-1$ samples are drawn from the pseudorandom source and the last $m-j$ samples are drawn from the uniform source. 

Let $v$ be the memory state of $A$ after seeing the first $j-1$ samples, which were generated from a pseudorandom source with a seed $x$ picked uniformly at random from $\zo^n$. Let $\P_{x|v}$ be the probability distribution of the random variable $x$ conditioned on reaching $v$. 
If the $j$th sample is generated using the same pseudorandom source, then $\forall \alpha \in [n]^{(k)}$, the bit $b$ is 0 with probability $\sum_{x': P(x'^\alpha)=0} \P_{x|v}(x')$ and 1 with probability $1-\sum_{x': P(x'^\alpha)=0} \P_{x|v}(x')$. If the $j$th sample is generated using the uniform source, then $\forall \alpha \in [n]^{(k)}$, the bit $b$ is 0 with probability $1/2$ and 1 with probability $1/2$. 
Thus, for any $\alpha$, $A$ can identify the ``unknown source" with an at most $\left|\sum_{x':P(x'^\alpha)=0} \P_{x|v}(x')-1/2\right|$ advantage. 

We show that when $A$ has low memory ($<n^\eps$ for some $0<\eps<1$), then with high probability, it reaches a state $v$ such that $\P_{x|v}$ has high min-entropy (informally, it's hard to determine the seed for the pseudorandom source). We then use $t$-resiliency of $P$ to show that when $\P_{x|v}$ has high min-entropy, then with high probability over $\alpha \in [n]^{(k)}$, $b$ behaves almost like in a uniform source (Lemma \ref{cl:spafourpred}), that is, $|\sum_{x':P(x'^\alpha)=0} \P_{x|v}(x')-1/2|$ is small. Hence, with high probability, it's hard for $A$ to judge with 'good' advantage whether $b$ was generated from a pseudorandom or a uniform source. Note that the last $m-j$ samples generated by a uniform source can't better this advantage.

\section{Time-Space Tradeoff through Reduction to Learning}
\label{sec:reduction}

In this section, we will prove time-space tradeoffs for the following distinguishing problems using black-box reduction from the 
corresponding learning problems.
\paragraph{Distinguishing Subspaces from Uniform}

Informally, we study the problem of distinguishing between the cases when the samples are drawn from a uniform distribution
over $\{0,1\}^n$ and when the samples are drawn randomly from a subspace of rank $k$ over $\F_2$.

Let $L(k,n)$ be the set of all linear subspaces of dimension $k$ ($\subseteq\{0,1\}^n$), that is, $L(k,n)=$
\begin{align*}
\left\{V\mid V=\{v\in\{0,1\}^n\mid \inner{w_i,v}=0\;\;\forall i\in[n-k]\} \text{ and where } w_1,w_2,...,w_{n-k} \text{ are linearly independent}\right\}
\end{align*}
Formally, we consider distinguishers for distinguishing between the following distributions:
\begin{enumerate}
\item $D_0$: Uniform distribution over $\{0,1\}^n$.
\item $D_1(S)$, $S\in L(k,n)$: Uniform distribution over $S$.  
\end{enumerate}

Note: If the subspace $S$ is revealed, it's easy for a branching program of constant width to distinguish w.h.p. by 
checking the inner product of the samples with a vector in the orthogonal complement of $S$. 

A distinguisher can distinguish subspaces if for an unknown random linear subspace $S\in L(k,n)$, 
it can distinguish between $D_0$ and $D_1(S)$. 
Formally, a distinguisher $L$, after seeing $m$ samples, has a success probability of $p$ if 
\begin{equation}\label{eq:s1}
\frac{\Pr_{u_1,...,u_m\sim D_0}[L(u_1,...,u_m)=0]+\Pr_{S\in_R L(k,n);u_1,...,u_m\sim D_1(S)}[L(u_1,...,u_m)=1]}{2}=p
\end{equation}

\begin{theorem}\label{thmmain:sub}
For $k>c_2\log n$ (where $c_2$ is a large enough constant), any algorithm that can distinguish $k$-dimensional subspaces over $\F_2^n$ from $\F_2^n$ ($\{0,1\}^n$), when samples are 
drawn uniformly at random from the subspace or $\F_2^n$ respectively, with success probability at least $\frac{1}{2}+2^{-o(k)}$ requires
either a memory of size $\Omega(k^2)$ or $2^{\Omega(k)}$ samples. 
\end{theorem}
We prove the theorem in Subsection \ref{sec:subspace}. Briefly, we prove that using a distinguisher for distinguishing subspaces, we can construct a branching program that learns an unknown bit vector $x$ from random linear equations over $\F_2$. Then, we are able to lift the time-space tradeoffs of Theorem \ref{thm:parity}.

\begin{remark}[\textbf{Tightness of the Lower Bound}] \label{rem:tightness-1}
We note two easy upper bounds that show that our results in Theorem \ref{thmmain:sub} are tight (up to constants in the exponent). Firstly, we observe an algorithm $B_1$ that distinguishes subspaces of dimension $k$ from uniform, using $O(k^2)$ memory and $O(k)$ samples, with probability at least $3/4$ ($0<k\le n-1$). $B_1$ stores the first $\min(8k,n)$ bits of the first $8k$ samples (in $O(k^2)$ memory); outputs 1 if the samples (projected onto the first $\min(8k,n)$ coordinates) belong to a $\le k$-dimensional subspace (of $\{0,1\}^{\min(8k,n)}$), and 0 otherwise (can be checked using gaussian elimination). When the samples are drawn from $D_1(S)$ for some $k$-dimensional subspace $S$, then $B_1$ always outputs the correct answer. When the samples are drawn from a uniform distribution on $\{0,1\}^n$, the probability that $8k$ samples form a $k$-dimensional subspace is at most 
$${8k\choose k}\cdot \frac{1}{2^{7k}}\le (8e)^k2^{-7k}<2^{-2k}\le 1/4$$ (because, if the $8k$ samples form a $k$-dimensional subspace, then at least $7k$ of them are linearly dependent on the previously stored samples and that happens with at most $1/2$ probability for each sample). Hence, $B_1$ errs with at most $1/4$ probability.

Secondly, we observe that there exists a branching program that distinguishes subspaces of dimension $k$ from uniform using constant width and $O(k\cdot 2^{k})$ length with probability at least $3/4$. Before, we show a randomized algorithm $P$ that distinguishes between $D_0$ and $D_1(S)$ for every $S\in L(k,n)$ with high probability. $P$ is described as follows:
\begin{enumerate}
\item\label{itm3:1} Repeat steps \ref{itm3:2} to \ref{itm3:3} sequentially for $t=10\cdot 2^k$ iterations.
\item\label{itm3:2} Pick a non-zero vector $v$ uniformly at random from $\{0,1\}^n$. For the next $2k$ samples (of the form $a\in \{0,1\}^n$), check if $\inner{a,v}=0$.
\item\label{itm3:3} If all the $2k$ samples are orthogonal to $v$, exit the loop and output $1$.
\item Output 0 (None of the randomly chosen vectors were orthogonal to all the samples seen in its corresponding iteration).
\end{enumerate}
The number of samples seen by $P$ is $20k\cdot 2^k$. Now, we prove that for every subspace $S$ of dimension $k$, that is, $S\in L(k,n)$, $P$ distinguishes between $D_0$ and $D_1(S)$ with probability at least $1-\frac{1}{2}(e^{-5}+\frac{10}{2^k})\ge 3/4$\footnote{$k\ge 5$}. 

When the samples are drawn from $D_0$, the probability that $P$ outputs 1 is equal to the probability that in at least one of the $t$ iterations, the randomly chosen non-zero vector $v$ was orthogonal to the $2k$ samples drawn uniformly from $\zo^{n}$. Here, the probability is over $v$ and the samples. By union bound, we can bound the probability of outputting $1$ (error) by 
$$10\cdot 2^k\cdot \left(\frac{1}{2}\right)^{2k}=\frac{10}{2^k}. $$ 
For a fixed subspace $S\in L(k,n)$, the probability that we pick a non-zero vector $v\in\zo^n$ that is orthogonal to $S$ is at least $\frac{2^{n-k}-1}{2^n-1}>2^{-(k+1)}$. Therefore, when the samples are drawn from $D_1(S)$, the probability that $P$ outputs 0 (error) is upper bounded by $\left(1-\frac{1}{2^{k+1}}\right)^{10\cdot 2^k}\le e^{-5}.$
Here, the probability is over the vectors $v$ and the samples. Now to construct a constant width but $20k\cdot 2^k$ length branching program that distinguishes with probability at least $3/4$, we consider a bunch of branching programs each indexed by $t$ vectors that are used in step \ref{itm3:2} of the algorithm $P$. It's easy to see that for a fixed set of $t$ vectors, $P$ can be implemented by a constant width branching program. As, when the $t$ vectors are uniformly distributed over $\zo^n$ (non-zero), $P$ can distinguish with probability at least $3/4$ for every subspace $S\in L(k,n)$, there exists a fixing to the $t$ vectors such that the corresponding branching program distinguishes between $D_0$ and $D_1(S)$ (when $S$ is chosen uniformly at random from $L(k,n)$) with probability at least $3/4$.
\end{remark}

\paragraph{Distinguishing Satisfiable Sparse Equations from Uniform}
Informally, we study the problem of distinguishing between the cases when the samples are drawn from 
satisfiable sparse equations over $\F_2$ and when the samples are drawn from random sparse equations.

Formally, we consider the distinguishing problem between the following two distributions:
\begin{enumerate}
\item $D_0$: Distribution on $(n+1)$-length vectors $(v^1,v^2,...,v^n,b)\in\{0,1\}^{n+1}$ where $\forall i\in[n], ~v^i$ is 1 
with probability $\frac{k}{n}$ and 0 otherwise, and $b$ is 1 with probability $\frac{1}{2}$ and 0 otherwise. 
\item $D_1(x)$, $x\in\{0,1\}^n$: Distribution on $(n+1)$-length vectors $(v^1,v^2,...,v^n,b)\in\{0,1\}^{n+1}$ where $\forall i\in[n], ~v^i$ is 1
 with probability $\frac{k}{n}$ and 0 otherwise, and $b=\inner{v,x}$ where $v=(v^1,v^2,...,v^n)$.
\end{enumerate}
Here, $k$ is the sparsity parameter. 

We say that a distinguisher can distinguish satisfiable sparse equations of sparsity $k$ from random ones if, 
when $x$ is unknown and chosen uniformly at random from $\{0,1\}^n$, it can distinguish between $D_0$ and $D_1(x)$.
Formally, a distinguisher $L$, after seeing $m$, has a success probability of $p$ if 
\begin{align}\label{eq:succprsparse}
\frac{\Pr_{u_1,...,u_m\sim D_0}[L(u_1,...,u_m)=0]+\Pr_{x\in_R\{0,1\}^n;u_1,...,u_m\sim D_1(x)}[L(u_1,...,u_m)=1]}{2}\ge p
\end{align}

\begin{theorem}\label{thm:superlinearparity}
For large enough $n$ and $k>c_5\log n$ (where $c_5$ is a large enough constant) and $k\le \frac{n}{4}$, any algorithm that can distinguish random sparse parities of sparsity $k$ on $n$ variables from satisfiable ones,
with success probability at least $\frac{1}{2}+2^{-o(k)}$, requires
either a memory of size $\Omega(nk)$ or $2^{\Omega(k)}$ samples. 
\end{theorem}
\begin{remark}[\textbf{Tightness of our Lower Bound}] \label{rem:tightness-2}
We note two easy upper bounds that show that our results in Theorem \ref{thm:superlinearparity} are almost tight. Firstly, we observe that there's an algorithm $B_1$ of memory $O(nk\log n)$ that uses $O(n)$ samples and can distinguish random sparse parities of sparsity $k$ from satisfiable ones, with probability of at least $3/4$. $B_1$ just stores $O(n)$ samples (in $O(nk\log n)$ memory); if there exists $x$ that satisfies all the samples, it outputs 1, otherwise it outputs 0. When the samples are satisfiable, that is, drawn from $D_1(x)$ (for some $x$), $B_1$ always outputs 1. When the samples are random, using the union bound, it's easy to see that the probability that there exists an $x$ that satisfies all the $O(n)$ samples is exponentially small in $n$. 

Second, there's an algorithm $B_2$ of constant memory that uses $O(n\cdot 2^{O(k)})$ samples and can distinguish random sparse parities of sparsity $k$ from satisfiable ones, with probability of at least $3/4$. The probability that a learning algorithm sees sample $(a,b)$, such that $a=(1,0,...,0)$, is at least $\frac{ke^{-2k}}{n}$ for $k<n/2$; thus, one can just wait for say $5$ such samples and see if the values of $b$ are drawn randomly or are fixed, giving a constant memory and $O(ne^{2k})$ samples algorithm that distinguishes with high probability. 
\end{remark}

The complete proof of Theorem \ref{thm:superlinearparity} is given in Section \ref{sec:supersparse}. Briefly, we prove that using such a distinguisher, we can construct a branching program that learns an unknown bit vector $x$ from sparse linear equations of sparsity $k$ over $\F_2$. Unlike before, when we were able to lift, we are not able to directly lift the time-space tradeoffs of Theorem \ref{thm:sparse}, because these lower bounds hold when the equations are of sparsity \emph{exactly}-$k$. Following the proof of Theorem \ref{thm:sparse}
in \cite{GRT18} very closely, we can prove the following lemma:

\begin{lemma}\label{lem:product}
Any branching program that learns $x\in\{0,1\}^n$ from random linear equations over $\F_2$ of type 
$(a,b)=(a^1,a^2,...,a^n,b)\in\{0,1\}^{n+1}$, where $\forall i\in[n],\; a^i=1$ with probability 
$\frac{k}{n}$ and $b=\inner{a,x}$, with success probability $2^{-ck}$, requires either 
width of size $2^{\Omega(n\cdot k)}$ or length of $2^{\Omega(k)}$
(where $c$ is a small enough constant, $k\le \frac{n}{4}$).
\end{lemma}
The proof is given in Section \ref{sec:supersparse}.
Therefore, through reduction as stated before, we are able to lift the time-space tradeoffs of Lemma \ref{lem:product} to get Theorem \ref{thm:superlinearparity}.

\subsection{Proof of Theorem \ref{thmmain:sub}}\label{sec:subspace}
\begin{proof} 
We will prove through reduction to Theorem \ref{thm:parity}. Let $B$ be the 
branching program that distinguishes subspaces of dimension $k$, with width $d$, length $m$ 
and success probability $\frac{1}{2}+s$. 
Using $B$, we show that there exists a branching program for parity learning over $\{0,1\}^{k'}$ (where $k<k'\le n$ and 
would be defined concretely below), with 
width $d2^{k'}(\frac{8n^2\log n}{s^2})^2$, length $mk'(\frac{8n^2\log n}{s^2})$ and success probability $1-\frac{1}{n}$. 
Hence, Theorem \ref{thm:parity} implies that either $d2^{k'}(\frac{8n^2\log n}{s^2})^2=2^{\Omega(k'^2)}$
or $mk'(\frac{8n^2\log n}{s^2})=2^{\Omega(k')}$. Assuming $s\ge 2^{-c_1k'}$ (where $c_1$ is a small enough constant),
$k\ge c_2\log n, c_3$ where $c_2, c_3$ are large enough constants, we get that $d=2^{\Omega(k'^2)}$ or $m=2^{\Omega(k')}$.
As $k'>k$: we have shown that if $B$ has success probability at least $\frac{1}{2}+2^{-c_1k}$ (for small enough constant $c_1$) at distinguishing
$k$-dimensional subspaces, $B$ has width at least $2^{\Omega(k^2)}$ or length $2^{\Omega(k)}$.\\

Firstly, using a simple argument, we show that $B$ can distinguish between 
subspaces of dimension $k'-1$ and $k'$ for some $k+1\le k'\le n$ with success probability
$\ge \frac{1}{2}+\frac{s}{n}$.  Writing the expression for success probability from Equation \ref{eq:s1},
\begin{align*}
&\frac{\Pr_{u_1,...,u_m\sim D_0}[B(u_1,...,u_m)=0]+\Pr_{S\in_R L(k,n);u_1,...,u_m\sim D_1(S)}[B(u_1,...,u_m)=1]}{2}&&=\frac{1}{2}+s \\
\implies&\Pr_{u_1,...,u_m\sim D_0}[B(u_1,...,u_m)=0]+1-\Pr_{S\in_R L(k,n);u_1,...,u_m\sim D_1(S)}[B(u_1,...,u_m)=0]&&=1+2s \\
\implies&\Pr_{u_1,...,u_m\sim D_0}[B(u_1,...,u_m)=0]-\Pr_{S\in_R L(k,n);u_1,...,u_m\sim D_1(S)}[B(u_1,...,u_m)=0]&&=2s
\end{align*}\\

The last expression on the left hand side can be written as
\[\sum_{k'=n}^{k+1}{\left(\Pr_{S\in_R L(k',n);u_1,...,u_m\sim D_1(S)}[B(u_1,...,u_m)=0]-\Pr_{S\in_R L(k'-1,n);u_1,...,u_m\sim D_1(S)}[B(u_1,...,u_m)=0]\right)}=2s\]
($D_1(S)=D_0$ for $S\in L(n,n)$ as $L(n,n)=\{\{0,1\}^n\}$)\\

Therefore, there exists $k+1\le k'\le n$ such that 
\begin{align}\label{eq:succpr}
\left(\Pr_{S\in_R L(k',n);u_1,...,u_m\sim D_1(S)}[B(u_1,...,u_m)=0]-\Pr_{S\in_R L(k'-1,n);u_1,...,u_m\sim D_1(S)}[B(u_1,...,u_m)=0]\right)\ge \frac{2s}{n}
\end{align}\\

We have shown that $B$  can solve the following distinguishing problem, that is, 
learn $\b$ with success probability at least $\frac{1}{2}+\frac{s}{n}$: If $\b=0$, then the distinguisher is given samples from a 
$k'$-dimensional subspace of $\{0,1\}^n$, otherwise (when $\b=1$), 
the distinguisher is given samples from a $(k'-1)$-dimensional subspace of $\{0,1\}^n$. 
Here, the probability is over $\b$, the $k'$ and $(k'-1)$-dimensional subspaces and the samples seen by $B$.\\

Next, using $B$, we construct a randomized learning algorithm $P$ for parity learning. 
The parity learning problem is as follows: a secret $x\in\{0,1\}^{k'}$ is chosen uniformly at random, 
the learner wants to learn $x$ from random linear equations over $\F_2$, that is, 
$(a,b)$ where $a\in_R\{0,1\}^{k'}$ and $b=\inner{a,x}$ ($\inner{a,x}$ is the inner product of $a$ and $x$ modulo 2). 
$P$ uses $B$ to guess each bit of $x$ one by one as follows:
\begin{enumerate}
\item For $i\in \{1,2,...,k'\}$, do Steps \ref{itm:start} to \ref{itm:end}.
\item\label{itm:start} Initiate $count_0=0, count_1=0$. These keep counts for the number of times the following algorithm outputs 
$0, 1$ respectively as the guess for $x^i$. 
\item \label{itm:2} Pick $g$ to be 0 with probability $\frac{1}{2}$ and 1 with probability $\frac{1}{2}$. This is a guess for $x^i$.

\item \label{itm:3} Let $\M$ be the set of all rank-$n$ linear maps $M:\{0,1\}^n\rightarrow\{0,1\}^n$ over $\F_2$, that is, the rows 
$\{M_r\}_{r\in[n]}$ are linearly independent. Pick $M\in\M$ uniformly at random. 

Let $f_M: \{0,1\}^{k'}\times\{0,1\}\rightarrow \{0,1\}^n$ be defined 
as $f_M(a,b)=M\cdot (a^{-i},b+ga^i,0,0,...0)$ (where $\cdot$ represents matrix-vector product, 
and $(a^{-i},b+ga^i)$ is appended with $n-k'$ zeroes).

For the next $m$ samples $(a_1,b_1),(a_2,b_2),...,(a_m,b_m)$, $P$ runs $B$ with $f_M(a_j,b_j)$ as $B$'s $j^{th}$ sample. 

\item  \label{itm:4} If $B(f_M(a_1,b_1),...,f_M(a_m,b_m))$ outputs $0$, then increase $count_{1-g}$ by 1, otherwise, increase $count_{g}$ by 1.
In the discussions below, we will see that we increase the count for $x^i$ with probability at least $(\frac{1}{2}+\frac{s}{n})$. 

\item\label{itm:end} Repeat steps \ref{itm:2} to \ref{itm:4} for $t=\frac{8n^2\log n}{s^2}$ times. If $count_0> count_1$, set $x'^i=0$ and 
store, else set $x'^i=1$ and store. As we will see below, $x'^i=x^i$ with probability at least $(1-\frac{1}{n^2})$.
\item Output $x'$ as the guess for $x$. 
\end{enumerate}

\begin{claim}\label{claim:succpr} For each $x\in\{0,1\}^{k'}$, if $x$ is the chosen secret, $P$ outputs $x'=x$ with probability at least $(1-\frac{1}{n})$.
\end{claim}
Here, the probability is over the samples, all the random guesses $g$ in Step \ref{itm:2} and the linear maps $M$ in Step \ref{itm:3}.
\begin{proof}
The probability that a single iteration of Steps \ref{itm:2} to \ref{itm:4} increases the counter for $x^i$ is
the probability that $B(f_M(a_1,b_1),...,f_M(a_m,b_m))$ outputs $1$ when $x^i=g$ and 0 when $x^i=1-g$. 

Consider the subspace $V_g=\{(a^{-i},b+ga^i)\in\{0,1\}^{k'}\mid a=(a^1,...,a^{k'})\in\{0,1\}^{k'}, b=\inner{a,x}\}$.
Here, the additions are modulo 2 and $a^{-i}\in\{0,1\}^{k'-1}$ is $a$ with the $i^{th}$ coordinate deleted. 
When $x^i=g$, $V_g$ forms a $(k'-1)$-dimensional subspace as $(x^{-i},1)$ is orthogonal to all the vectors in $V_g$.
As $M$ is full rank, the range of $f_M(a,b)$ forms a $(k'-1)$-dimensional
subspace too and under $M$ being picked uniformly at random from $\M$, we get a uniform distribution on the $(k'-1)$-dimensional
subspaces ($L(k'-1,n)$). 
When $x^{i}\not =g$, it's easy to see that $V_g=\{0,1\}^{k'}$ and thus, Range$(f_M)$ under 
$M\in_R\M$ is a uniform distribution on $L(k',n)$. 
Therefore,
\begin{align*}
&\Pr_{\substack{g\in_R\{0,1\};M\in_R\M;\\a_1,a_2,...,a_{m}\in_R\{0,1\}^{k'};\forall j\in[m], b_j=\inner{a_j,x}}}{[B(f_M(a_1,b_1),...,f_M(a_m,b_m))=1\;\wedge \;x^i=g]}\\
&\;\;\;\;+\Pr_{\substack{g\in_R\{0,1\};M\in_R\M;\\a_1,a_2,...,a_{m}\in_R\{0,1\}^{k'};\forall j\in[m], b_j=\inner{a_j,x}}}{[B(f_M(a_1,b_1),...,f_M(a_m,b_m))=0\;\wedge \; x^i=1-g]}\\
=&\;\frac{1}{2}\left(\Pr_{S\in_RL(k'-1,n); u_1,u_2,...,u_m\in_RS}{[B(u_1,...,u_m)=1]}+\Pr_{S\in_RL(k',n); u_1,u_2,...,u_m\in_RS}{[B(u_1,...,u_m)=0]}\right)\\
=&\;\frac{1}{2}\left(1+\Pr_{S\in_RL(k',n); u_1,u_2,...,u_m\in_RS}{[B(u_1,...,u_m)=0]}-\Pr_{S\in_RL(k'-1,n); u_1,u_2,...,u_m\in_RS}{[B(u_1,...,u_m)=0]}\right)\\
\ge& \;\frac{1}{2}+\frac{s}{n} 
\end{align*}
The last inequality follows from Equation \ref{eq:succpr}.\\

Next, we prove that $x'^i=x^i$ with probability at least $(1-\frac{1}{n^2})$ using Chernoff Bound. For $o=1$ to $t$,
let $X_o=1$ if we increase $count_{x^i}$ in the $o^{th}$ iteration of Steps \ref{itm:2} to \ref{itm:4} for calculating $x^i$, and 0 otherwise.
From the previous argument, we know that $\Ex(X_o)\ge  \frac{1}{2}+\frac{s}{n} $. As, $\{X_o\}_{o\in[t]}$ are independent 
random variables. 
\begin{align*}
\Pr[x'^i\not =x^i]=\Pr\left[\sum_{o=1}^t{X_o}\le\frac{t}{2}\right]&\le \Pr\left[\sum_{o=1}^t{X_o}-\Ex({\sum_{o=1}^t{X_o}})\le-\frac{ts}{n}\right]\\
&\le e^{-\frac{t}{4}(\frac{s}{n})^2}\le \frac{1}{n^2}
\end{align*}\\

Claim \ref{claim:succpr} just follows from union bound, that is, 
\begin{align*}
\Pr[x'\not= x]\le \sum_{i=1}^{k'}\Pr[x'^i\not= x^i]\le k'\left(\frac{1}{n^2}\right)\le \frac{1}{n}
\end{align*}

\end{proof}
Using $P$, we construct a set of branching programs one for each possible set of guesses $g$ 
and linear maps $M$. Let $P\left[\{g^i_o\}_{i\in[k'],o\in[t]},\{M^i_o\}_{i\in[k'],o\in[t]}\right]$ represent a branching program
that executes $P$ with $g^i_o$ as the guess for $x^i$ and $M^i_o$ as the linear map in the $o^{th}$ iteration of Step \ref{itm:2} to \ref{itm:4}
for calculating $x^i$.\\

$P\left[\{g^i_o\}_{i\in[k'],o\in[t]},\{M^i_o\}_{i\in[k'],o\in[t]}\right]$ can run $B$ on modified samples in Step \ref{itm:3} using the same width as $B$,
as after fixing the linear map, each modified sample depends only on a single sample seen by $P$. And because a branching program is 
a non-uniform model of computation, $P\left[\{g^i_o\}_{i\in[k'],o\in[t]},\{M^i_o\}_{i\in[k'],o\in[t]}\right]$ doesn't need to store the guesses
and maps. 
It does need to store $x'$, $count_0$, $count_1$ in addition to the width of $B$,
where the space for counts is reused for each $i$. Therefore, the width ($d'$) of the branching programs, based on $P$, 
is $\le d2^{k'}(2^{\log t})^2=d2^{k'}(\frac{8n^2\log n}{s^2})^2$.\\

It is easy to see that the length ($m'$) of $P\left[\{g^i_o\}_{i\in[k'],o\in[t]},\{M^i_o\}_{i\in[k'],o\in[t]}\right]$ is $mk't=mk'(\frac{8n^2\log n}{s^2})$.\\

Through Claim \ref{claim:succpr}, we know that for each $x$,
\[\Pr_{\substack{g^i_o\in_R\{0,1\};M^i_o\in_R\M;\\a_1,a_2,...,a_{m'}\in_R\{0,1\}^{k'};\forall j\in[m'], b_j=\inner{a_j,x}}}[P((a_1,b_1),...,(a_{m'},b_{m'}))=x]\ge1-\frac{1}{n} \]

Therefore, 
\[\Pr_{\substack{x\in\{0,1\}^{k'};g^i_o\in_R\{0,1\};M^i_o\in_R\M;\\a_1,a_2,...,a_{m'}\in_R\{0,1\}^{k'};\forall j\in[m'], b_j=\inner{a_j,x}}}[P((a_1,b_1),...,(a_{m'},b_{m'}))=x]\ge1-\frac{1}{n} \]

The above expression  can be rewritten as follows:
\[\Ex_{g^i_o\in_R\{0,1\};M^i_o\in_R\M}\left(\Pr_{\substack{x\in\{0,1\}^{k'};a_1,a_2,...,a_{m'}\in_R\{0,1\}^{k'};\\\forall j\in[m'], b_j=\inner{a_j,x}}}[P((a_1,b_1),...,(a_{m'},b_{m'}))=x]\right)\ge1-\frac{1}{n} \]

Therefore, there exist guesses $\{g^i_o\}_{i\in[k'],o\in[t]}$ and linear maps $\{M^i_o\}_{i\in[k'],o\in[t]}$ such that 
\[\Pr_{\substack{x\in\{0,1\}^{k'};a_1,a_2,...,a_{m'}\in_R\{0,1\}^{k'};\\\forall j\in[m'], b_j=\inner{a_j,x}}}\left[P\left[\{g^i_o\}_{i\in[k'],o\in[t]},\{M^i_o\}_{i\in[k'],o\in[t]}\right]((a_1,b_1),...,(a_{m'},b_{m'}))=x\right]\ge 1-\frac{1}{n}.\]\\

This gives us a branching program of width $d'$ and length $m'$ for parity learning with success probability at least $1-\frac{1}{n}$.\\
\end{proof}

\subsection{Proof of Theorem \ref{thm:superlinearparity}}\label{sec:supersparse}

\begin{proof} We prove through reduction to Lemma \ref{lem:product}.
Let $B$ be the 
branching program with width $d$ and length $m$ that distinguishes random linear equations of 
sparsity $k$ on $n$ variables from satisfiable ones, with success probability $(\frac{1}{2}+s)$.

Using $B$, we show that there exists a branching program for learning $x\in \{0,1\}^{n+1}$ from random
linear equations $(a,b)=(a^1,a^2,...,a^{n+1},b)\in\{0,1\}^{n+2}$, where $\forall i\in[n+1],\; a^i=1$ with probability 
$\frac{k'}{n+1}$ and $b=\inner{a,x}$, with 
width $d2^{n+1}(\frac{16\log n}{s^2})^2$, length $n^2\left(m+16\log(\frac{n}{s})\right)\left(\frac{16\log n}{s^2}\right)$ 
and success probability $1-\frac{1}{n}$. 
Here $k'=\frac{k(n+1)}{n}$. Let $n'=n+1$.
Hence, Lemma \ref{lem:product} implies that either $d2^{n+1}(\frac{16\log n}{s^2})^2=2^{\Omega(n'k')}=2^{\Omega(nk)}$
or $n^2\left(m+16\log(\frac{n}{s})\right)\left(\frac{16\log n}{s^2}\right)=2^{\Omega(k')}=2^{\Omega(k)}$. 
Assuming $s\ge 2^{-c_4k}$ (where $c_4$ is small enough constant),
$k\ge c_5\log n, c_6$ where $c_5, c_6$ are large enough constants, we get that $d=2^{\Omega(nk)}$ or $m=2^{\Omega(k)}$.
Therefore, if $B$ has success probability of at least $\frac{1}{2}+2^{-ck}$ (for small enough constant $c$), 
$B$ has width $2^{\Omega(nk)}$ or length $2^{\Omega(k)}$.\\


Next, we construct a randomized learning algorithm $P$ that learns $x$ from random sparse linear equations. 
Reiterating, the problem is as follows: a secret $x\in\{0,1\}^{n'}$ is chosen uniformly at random, 
the learner wants to learn $x$ from random linear equations of sparsity $k'$ over $\F_2$, that is, 
$(a,b)=(a^1,a^2,...,a^{n'},b)$ where $\forall i\in[n'],\; a^i=1$ with probability $\frac{k'}{n'}=\frac{k}{n}$ and 
$b=\inner{a,x}$ ($\inner{a,x}$ is the inner product of $a$ and $x$ modulo 2). 
$P$ uses $B$ to guess each bit of $x$ one by one as follows:
\begin{enumerate}
\item For $i\in \{1,2,...,n'\}$, do Steps \ref{itm2:start} to \ref{itm2:end}.
\item\label{itm2:start} Initiate $count_0=0, count_1=0$. These keep counts for the number of times the following algorithm outputs 
$0, 1$ respectively as the guess for $x^i$. 
\item \label{itm2:2} Pick $g$ to be 0 with probability $\frac{1}{2}$ and 1 with probability $\frac{1}{2}$. $g$ is a guess for $x^i$. 
Pick a random $y\in\{0,1\}^n$.
Consider the following map $f_y:\{0,1\}^{n+2}\rightarrow\{0,1\}^{n+1}$ defined as $f_y(a,b)=(a^{-i},b+ga^i+\inner{a^{-i},y})$ 
($a\in\{0,1\}^{n+1},b\in\{0,1\}$).

\item \label{itm2:3} $P$ uses the branching program $B$ to check if it's guess $g$ is correct. 
For the next $\mm=\frac{n}{k}\left(m+16\log(\frac{n}{s})\right)$ samples $(a_1,b_1),(a_2,b_2),...,(a_{\mm},b_{\mm})$, $\forall j\in [\mm]$, if $a_j^i=1$, 
$P$ feeds $f_y(a_j,b_j)$ as $B$'s next sample. If $a_j^i=0$, with probability $\frac{k}{n-k}$, 
$P$ feeds $f_y(a_j,b_j)$ as $B$'s next sample, otherwise, with probability $1-\frac{k}{n-k}$,
$P$ throws away the sample. We will show that with high probability (at least $1-(e^{-\frac{m}{4}})\frac{s^4}{n^4}$),
$B$ is run over at least $m$ samples (the probability is over the randomness of the samples). If $P$ feeds less than 
$m$ samples to $B$, halt the procedure and output $0^{n'}$ as the guess for $x$. 

\item  \label{itm2:4} If $B$ outputs $0$, then increase $count_{1-g}$ by 1, otherwise, increase $count_{g}$ by 1.
In the discussions below, we will see that we increase the count for $x^i$ with probability at least $(\frac{1}{2}+s)$
when the procedure didn't halt. 

\item\label{itm2:end} Repeat steps \ref{itm2:2} to \ref{itm2:4} for $t=\frac{16\log n}{s^2}$ times. If $count_0> count_1$, set $x'^i=0$ and 
store, else set $x'^i=1$ and store. As we will see below, $x'^i=x^i$ with probability at least $(1-\frac{1}{n^4})$.
\item Output $x'$ as the guess for $x$. 
\end{enumerate}

\begin{claim}\label{claim:succprsparse} For all $x\in\{0,1\}^{n'}$, if the secret is $x$, $P$ outputs $x'=x$ with probability at least $(1-\frac{1}{n})$.
\end{claim}
Here, the probability is over the samples, all the random guesses $g$ in Step \ref{itm2:2} and the random $n$-bit vectors $y$ in Step \ref{itm2:2}.
\begin{proof}
The procedure halts when $P$ does not feed at least $m$ samples to $B$ in Step \ref{itm2:3}. 
\begin{align*}
\Pr_{(a,b)\sim D_1(x)}[(a,b) \text{ is used to generate $B$'s next sample} ]\\
=\Pr_{(a,b)\sim D_1(x)}[a^i=1]+\frac{k}{n-k}\left(\Pr_{(a,b)\sim D_1(x)}[a^i=0]\right)=\frac{2k}{n}
\end{align*}
After seeing $\mm$ samples, $P$ feeds $2\left(m+16\log (\frac{n}{s})\right)$ samples to $B$ in expectation. 
Therefore, using Chernoff bound, the probability that the procedure halts in Step \ref{itm2:3} is bounded by
$e^{-\frac{m+16\log (\frac{n}{s})}{4}}=e^{-\frac{m}{4}}(\frac{s}{n})^4$.

The probability that a single iteration of Steps \ref{itm2:2} to \ref{itm2:4} increases the counter for $x^i$ is
the probability that $B$ outputs $1$ when $x^i=g$ and 0 when $x^i=1-g$. The samples seen by $P$ ($(a,b)=(a^1,...,a^{n'},b)$)
are drawn such that $\forall h\in[n'],\; a^h=1$ with probability $\frac{k}{n}$ and 0 otherwise and $b=\inner{a,x}$.

It's easy to see that the distribution over the samples conditioned on them being fed to $B$ is as follows:
$\forall h\in[n'] : h\not =i,\; a^h=1$ with probability $\frac{k}{n}$ and 0 otherwise, 
$a^i=1$ or $0$ with probability $\frac{1}{2}$ each and $b=\inner{a,x}$.
We show that when $x^i=g$, $\forall y$, the distribution over the samples $f_y(a,b)\in\{0,1\}^{n+1}$, 
that $B$ sees, is equivalent to
$D_1(x^{-i}+y)$ whereas when $x^i=1-g$,  the distribution over the samples $f_y(a,b)$, that $B$ sees, is equivalent to
$D_0$.

When $x^i=g$, $\forall h\in[n]$, $f_y(a,b)^h=(a^{-i})^h$ which is 1 with probability $\frac{k}{n}$ and 0 otherwise. 
$f_y(a,b)^{n+1}=\inner{a,x}+ga^i+\inner{a^{-i},y}=\inner{a^{-i},x^{-i}}+(x^i+g)a^i+\inner{a^{-i},y}=\inner{a^{-i},x^{-i}+y}$. 
This is equivalent to the distribution 
$D_1(x^{-i}+y)$.

When $x^i=1-g$, $\forall h\in[n]$, $f_y(a,b)^h=(a^{-i})^h$ which is 1 with probability $\frac{k}{n}$ and 0 otherwise. 
$f_y(a,b)^{n+1}=\inner{a,x}+ga^i+\inner{a^{-i},y}=\inner{a^{-i},x^{-i}}+(x^i+g)a^i+\inner{a^{-i},y}=\inner{a^{-i},x^{-i}+y}+a^i$. 
As $a^i$ is 0 and 1 with $\frac{1}{2}$
and is independent of $a^{-i}$, this distribution is equivalent to the distribution $D_0$.

Therefore, 
\begin{align*}
&\Pr_{\substack{g\in_R\{0,1\}; y\in_R\{0,1\}^n; \\(a_1,b_1),(a_2,b_2),...,(a_{\mm},b_{\mm})\sim D_1(x)}}{[B \text{ outputs } 1\;\wedge \;x^i=g]}\\
&\;\;\;\;\;\;\;\;\;
+\Pr_{\substack{g\in_R\{0,1\}; y\in_R\{0,1\}^n; \\(a_1,b_1),(a_2,b_2),...,(a_{\mm},b_{\mm})\sim D_1(x)}}{[B \text{ outputs } 0\;\wedge \; x^i=1-g]}\\
=&\;\frac{1}{2}\left(\Pr_{ y\in_R\{0,1\}^n; u_1,u_2,...,u_m\sim D_1(x^{-i}+y)}{[B(u_1,...,u_m)=1]}+\Pr_{y\in_R\{0,1\}^n;u_1,u_2,...,u_m\sim D_0}{[B(u_1,...,u_m)=0]}\right)\\
=&\;\frac{1}{2}\left(\Pr_{ y'\in_R\{0,1\}^n; u_1,u_2,...,u_m\sim D_1(y')}{[B(u_1,...,u_m)=1]}+\Pr_{u_1,u_2,...,u_m\sim D_0}{[B(u_1,...,u_m)=0]}\right)\\
\ge& \;\frac{1}{2}+s
\end{align*}
The last inequality follows from Equation \ref{eq:succprsparse}.

Next, we prove that $x'^i=x^i$ with probability at least $(1-\frac{1}{n^4})$ using Chernoff Bound 
when the procedure does not halt. For $o=1$ to $t$,
let $X_o=1$ if we increase $count_{x^i}$ in the $o^{th}$ iteration of Steps \ref{itm2:2} to \ref{itm2:4} 
for calculating $x^i$, and 0 otherwise.
From the previous argument, we know that $\Ex(X_o)\ge  \frac{1}{2}+s $. As, $\{X_o\}_{o\in[t]}$ are independent 
random variables. 
\begin{align*}
\Pr[x'^i\not =x^i]=\Pr\left[\sum_{o=1}^t{X_o}\le\frac{t}{2}\right]&\le \Pr\left[\sum_{o=1}^t{X_o}-\Ex({\sum_{o=1}^t{X_o}})\le-ts\right]\\
&\le e^{-\frac{t}{4}s^2}\le \frac{1}{n^4}
\end{align*}

Claim \ref{claim:succprsparse} just follows from union bound (for $n$ greater than a large enough constant), that is, 
\begin{align*}
\Pr[x'\not= x]\le \sum_{i=1}^{n'}\Pr[x'^i\not= x^i]+\Pr[\text{procedure halts}]\le n'\left(\frac{1}{n^4}+te^{-\frac{m}{4}}\frac{s^4}{n^4}\right)\le \frac{1}{n}
\end{align*}
\end{proof}

Using $P$, we construct a set of branching programs one for each possible set of guesses $g$, $n$-bit vectors $y$ and choice
of throwing away the samples. 
Let $P\left[\{g^i_o\}_{i\in[n'],o\in[t]},\{y^i_o\}_{i\in[n'],o\in[t]},\{c^i_{o,q}\}_{i\in[n'],o\in[t],q\in[\mm]}\}\right]$ represent a branching program
that executes $P$ with $g^i_o$ as the guess for $x^i$ and $y^i_o$ as the seed for the map $f$ 
in the $o^{th}$ iteration of Step \ref{itm2:2} to \ref{itm2:4}
for calculating $x^i$. And $\forall\; i,o,q$, if $c^i_{o,q}=1$, then $P$ throws away whenever $P$ needs to decide whether to throw 
away the $q^{th}$ sample in the the $o^{th}$ iteration of Step \ref{itm2:2} to \ref{itm2:4}
for calculating $x^i$ and does not throw away if $c^i_{o,q}=0$. (Note that, in the learning algorithm $P$, $c^i_{o,q}$ is drawn from 
$Ber(1-\frac{k}{n-k})$)\\

$P\left[\{g^i_o\}_{i\in[n'],o\in[t]},\{y^i_o\}_{i\in[n'],o\in[t]},\{c^i_{o,q}\}_{i\in[n'],o\in[t],q\in[\mm]}\}\right]$ can run $B$ on 
modified samples in Step \ref{itm2:3} using the same width as $B$
as each modified sample depends only on a single sample seen by $P$. 
And because a branching program is 
a non-uniform model of computation, $P\left[\{g^i_o\}_{i\in[n'],o\in[t]},\{y^i_o\}_{i\in[n'],o\in[t]},\{c^i_{o,q}\}_{i\in[n'],o\in[t],q\in[\mm]}\}\right]$ 
doesn't need to store the guesses, random vectors $y$ and choices. 
It does need to store $x'$, $count_0$, $count_1$ in addition to the width of $B$,
where the space for counts is reused for each $i$. Therefore, the width ($d'$) of the branching programs, based on $P$, 
is $\le d2^{n'}(2^{\log t})^2=d2^{n'}(\frac{16\log n}{s^2})^2$.\\

It is easy to see that the length ($m'$) of $P\left[\{g^i_o\}_{i\in[n'],o\in[t]},\{y^i_o\}_{i\in[n'],o\in[t]},\{c^i_{o,q}\}_{i\in[n'],o\in[t],q\in[\mm]}\}\right]$ 
is \[\mm n't\le n^2\left(m+16\log(\frac{n}{s})\right)\left(\frac{16\log n}{s^2}\right).\]

Through Claim \ref{claim:succprsparse}, we know that for all $x\in \{0,1\}^{n'}$,
\[\Pr_{\substack{g^i_o;y^i_o;c^i_{o,q};\\(a_1,b_1),(a_2,b_2),...,(a_{m'},b_{m'})\sim D_1(x)}}[P((a_1,b_1),...,(a_{m'},b_{m'}))=x]\ge1-\frac{1}{n} \]

Therefore, 
\begin{align*}
\Pr_{\substack{g^i_o;y^i_o;c^i_{o,q};x\in\{0,1\}^{n'};\\(a_1,b_1),(a_2,b_2),...,(a_{m'},b_{m'})\sim D_1(x)}}
[P((a_1,b_1),...,(a_{m'},b_{m'}))=x]&\ge 1-\frac{1}{n}
\end{align*}

The above expression  can be rewritten as follows:
\[\Ex_{g^i_o;y^i_o;c^i_{o,q}}\left(\Pr_{x\in\{0,1\}^{n'};(a_1,b_1),(a_2,b_2),...,(a_{m'},b_{m'})\sim D_1(x)}
[P((a_1,b_1),...,(a_{m'},b_{m'}))=x]\right)\ge1-\frac{1}{n} \]

Therefore, there exists guesses $\{g^i_o\}$, $n$-bit vectors $\{y^i_o\}$, and choices $\{c^i_{o,q}\}$ such that 
\begin{align*}
\Pr_{x\in\{0,1\}^{n'};(a_1,b_1),...,(a_{m'},b_{m'})\sim D_1(x)}
\left[P\left[\{g^i_o\},\{y^i_o\},\{c^i_{o,q}\}\right]((a_1,b_1),...,(a_{m'},b_{m'}))=x\right]
\ge 1-\frac{1}{n}
\end{align*}

This gives us a branching program of width $d'$ and length $m'$ for learning from random linear equations of sparsity $k$ with success probability at least $1-\frac{1}{n}$.\\

\end{proof}

We now complete the proof of Lemma \ref{lem:product}.

\begin{proof}
Let $L$ be the branching program that learns $x\in\{0,1\}^n$ from $m$ random linear equations over $\F_2$ of type 
$(a,b)=(a^1,a^2,...,a^n,b)\in\{0,1\}^{n+1}$, where $\forall i\in[n],\; a^i=1$ with probability $\frac{k}{n}$ and $b=\inner{a,x}$, 
with success probability $s$. The success probability is over $x$ being uniformly picked at random from $\{0,1\}^n$ and over 
the $m$ equations.\\

First, we show that $L$ learns $x\in\{0,1\}^n$, with success 
probability at least $s-2m\cdot2^{-\frac{k}{8}}$, when the samples $(a,b)=(a^1,a^2,...,a^n,b)\in\{0,1\}^{n+1}$ are 
drawn from the same distribution but 
conditioned on the number of 1s in $a$ lying in the interval $[\frac{k}{2},2k]$. 

Let $p$ be the probability that number of 1s in $a=(a^1,a^2,...,a^n)$ are less that $\frac{k}{2}$ or 
greater than $2k$, when each $a^{i}$ is drawn independently from Ber$(\frac{k}{n})$. 
Using Chernoff Bound, it's easy to see that $p\le e^{-\frac{k}{8}}+e^{-\frac{k}{3}}\le 2^{-\frac{k}{8}}$ 
($k$ is greater than a large enough constant).
 
Let $P_0^x$ represent the distribution over $(n+1)$-length vectors $(a,b)=(a^1,a^2,...,a^n,b)$ 
where $\forall i\in[n],\; a^i=1$ with probability $\frac{k}{n}$ and $b=\inner{a,x}$. Let $P_1^x$ represent
 the distribution over $(n+1)$-length vectors $(a,b)=(a^1,a^2,...,a^n,b)$ 
where $\forall i\in[n],\; a^i=1$ with probability $\frac{k}{n}$ but conditioned on the number of 1s in $a$ being at least $\frac{k}{2}$
and at most $2k$, and $b=\inner{a,x}$.

\begin{align*}
&\left|\Pr_{x\in_R\{0,1\}^n,u_1,...,u_m\sim P_0^x}[L(u_1,...,u_m)=x]-\Pr_{x\in_R\{0,1\}^n,u_1,...,u_m\sim P_1^x}[L(u_1,...,u_m)=x]\right|\\
&\le m\cdot \max_x\left\{\sum_{(a,b)\in\{0,1\}^{n+1}}{|P_0^x(a,b)-P_1^x(a,b)|}\right\} \le m\cdot 2p \le 2m\cdot 2^{-\frac{k}{8}}
\end{align*}
Therefore, $L$ learns $x\in\{0,1\}^n$ from $m$ independent samples drawn from $P_1^x$ with success probability 
at least $s-2m\cdot 2^{-\frac{k}{8}}$.\\

Next, using the techniques from \cite{KRT17,GRT18}, we show a time-space tradeoff for such a branching program $L$. 
Let $T_l=\{a\in\{0,1\}^n\; :\; \sum_{i=1}^n{a^i}=l\}$. Let $M_l:T_l\times\{0,1\}^n\rightarrow \{-1,1\}$ be the matrix 
such that $M_l(a,x)=(-1)^{\inner{a,x}}$. Lemmas from \cite{GRT18} (Lemma 5.8 and 5.10) show that $M_l$ is 
a $(c_1l,c_1n)-L_2-$Extractor with error $2^{-c_1l}$ ($c_1>0$ is a sufficiently small constant and $l\le \frac{n}{2})$. 

Let $M:\bigcup_{l\in[\frac{k}{2},k]}{T_l}\times \{0,1\}^n\rightarrow\{-1,1\}$ be the matrix such that 
$M(a,x)=(-1)^{\inner{a,x}}$. Given $x$, the learning algorithm $L$ gets samples from the distribution $P_1^x$. Let $P_1^{(x,l)}$ 
represent the distribution over $(n+1)$-length vectors $(a,b)$ 
where $a$ is drawn uniformly at random from $T_l$ and $b=\inner{a,x}$. It's easy to see that
$P_1^x$ is a convex combination of the distributions $P_1^{(x,l)}, \; l\in[\frac{k}{2},k]$. 
As $M_l$ is a $(\frac{c_1k}{2},c_1n)-L_2-$Extractor with error $2^{-\frac{c_1k}{2}}$ for all $l\in[\frac{k}{2},k]$, 
it easily follows that for every non-negative $f:\{0,1\}^n\rightarrow \R$ with $\frac{\norm{f}_2}{\norm{f}_1} \le 2^{c_1n}$ 
the set of rows $a$ in $\bigcup_{l\in[\frac{k}{2},k]}{T_l}$ with
\[
\frac{|\inner{M_a,f}|}{\norm{f}_1}
\ge 2^{-\frac{c_1k}{2}}
\]
 has probability mass of at most $2^{-\frac{c_1k}{2}}$ under the following distribution: $\forall i\in[n],\; a^i=1$ with probability $\frac{k}{n}$ but 
conditioned on the number of 1s in $a$ lying in the interval $[\frac{k}{2},2k]$ (where $a=(a^1,a^2,...,a^n)$). 
Let $M':\left(\bigcup_{l\in[\frac{k}{2},k]}{T_l}\times\{0,1\}\right)\times \{0,1\}^n\rightarrow[0,1]$ be defined as follows: $M'((a,b),x)=P_1^x(a)\cdot \one_{\inner{a,x}=b}$. (Here, $\one_{\inner{a,x}=b}=1$ if $\inner{a,x}=b$ and 0 otherwise).

The above mentioned property of $M$ implies that $M'$ is a $(\frac{c_1k}{2},c_1n,1)-L_2-$Extractor with error $2^{-\frac{c_1k}{2}}$ according to the definition in \cite{GRT18} (Definition 6.1). Thus, a theorem from \cite{GRT18} (Theorem 9) allows us to prove that any branching program that learns
$x$ through independent samples drawn from $P_1^x$, with success probability $2^{-c_2 k}$ requires either memory of size 
$c_2nk$ or $2^{c_2 k}$ samples. Here, $c_2$ is a small enough constant and $k\le \frac{n}{4}$.\\

This proves the lemma as we already showed that $L$ learns $x$ through independent samples drawn from $P_1^x$, 
with success probability $s-2m\cdot 2^{-\frac{ k}{8}}$. Therefore, for $s\ge 2^{-c_3k}$, $m\le 2^{c_3k}$, $L$ should have memory of
size at least $c_2n k$ (where $c_3$ is a small enough constant compared to $c_2$).

\end{proof}

\section{Sample-Memory Tradeoffs for Resilient Local PRGs}
\label{sec:sublinear}

In this section, we prove our lower bound against memory bounded algorithms for distinguishing between streaming outputs of Goldreich's pseudorandom generator and perfectly random bits. 

Before stating and proving our result in detail, we set up some notation and definitions that will be convenient for us in this section. 

\subsection{Formal Setup}

A $k$-ary \emph{predicate} $P$ is a Boolean function $P:\zo^k \rightarrow \zo$. 
Let $ \sum_{\alpha \subseteq [k]} \hat{P}(\alpha) \chi_{\alpha}$ be the Fourier polynomial for $(-1)^P$ ($(-1)^P(x)=(-1)^{P(x)}$). 
$P$ is said to be $t$\emph{-resilient} if $t$ is the maximum positive integer such that $\hat{P}(\alpha) = 0$ whenever $|\alpha|< t$. 
In particular, the parity function $\inner{\alpha,x}$ is $|\alpha|$-resilient. Here,  $\chi_{\alpha}:\zo^k\rightarrow \{-1,1\}$ is such that 
$\chi_\alpha(x)=(-1)^{\inner{\alpha,x}}$.

Let $[n]^{(k)}$ denote the set of all ordered $k$-tuples of exactly $k$ elements of $[n]$. 
That is, no element of $[n]$ occurs more than once in any tuple of $[n]^{(k)}$. 
For any $a \in [n]^{(k)}$, let $a^i\in [n]$ denote the element of $[n]$ appearing in the $i$th position in $a$.
Given $x \in \zo^n$ and $a \in [n]^{(k)}$, let $x^a \in \zo^k$ be defined so that $(x^a)^i = x^{a^i}$ for every $1 \leq i \leq k$.

For any $k$-ary predicate $P$, consider the problem of distinguishing between the following two distributions on $(a,b) \in [n]^{(k)}\times\{0,1\}$ where $(a,b)$ are sampled as follows: 

\begin{enumerate}
\item $D_{null}$: 1) Choose $a$ uniformly at random from $[n]^{(k)}$, and 2) choose $b$ uniformly at random and independently from $\zo$.  
\item $D_{planted}(x)$, $x\in\{0,1\}^n$: 1) Choose $a$ uniformly at random from $[n]^{(k)}$, and 2) set $b = P(x^a)$. 
\end{enumerate}

Note that  $a$ is chosen uniformly at random from $[n]^{(k)}$ in both distributions. However, while the bit $b$ is independent of $a$ in $D_{null}$, it may be correlated with $a$ in $D_{planted}$. 

A distinguisher for the above problem gets access to $m$ i.i.d. samples $u_t=(a_t,b_t), t\in[m]$ from one of $D_{null}$ and $D_{planted}(x)$ for a uniformly randomly chosen $x \in \zo^n$ and outputs either ``planted'' or ``null''.  We say that the distinguisher succeeds with probability $p$ if:

\begin{align*}
\frac{\Pr_{u_1,...,u_m\sim D_{null}}[L(u_1,...,u_m)= \text{ ``null''}]+\Pr_{\substack{x \in_R \zo^n; \\u_1,...,u_m\sim D_{planted}(x)}}[L(u_1,...,u_m)= \text{``planted''}]}{2}\ge p
\end{align*}

Note: In the language used in the previous sections, think of ``null" as being equivalent to 0 and ``planted" being equivalent to 1, that is, $D_{null}\equiv D_0$ and $D_{planted}(x)\equiv D_1(x)$. Therefore, the success probability of the distinguisher $L$ can be written as 
\begin{equation}\label{eq:succprexactpred}
\frac{\Pr_{u_1,...,u_m\sim D_{0}}[L(u_1,...,u_m)=0]+\Pr_{x\in_R  \zo^n; u_1,...,u_m\sim D_{1}(x)}[L(u_1,...,u_m)= 1]}{2}\ge p
\end{equation}

In particular, if $x \in \zo^n$ is ``revealed'' to a distinguishing algorithm, then it is easy to use $\Theta(\log(1/\epsilon))$ samples and constant width branching program to distinguish correctly with probability at least $1-\epsilon$ between $D_{null}$ and $D_{planted}$. 

\subsection{Main Result}
The main result of this section is the following sample-memory trade-off for any distinguisher:
\begin{theorem}\label{thm:pred}
Let $P$ be a $t$-resilient $k$-ary predicate. Let $0<\epsilon<1-3\frac{\log 24}{\log n}$ and $k< n/c$. Suppose there's an algorithm that distinguishes between $D_{null}$ and $D_{planted}$ with probability at least $1/2 +s$ and uses $<n^{\epsilon}$ memory. Then, whenever $0<t \le k< n^{\frac{(1-\epsilon)}{6}}/3$ and $s >c_1 ( \frac{n}{t})^{-(\frac{1-\eps}{36})t}$, the algorithm requires $ (\frac{n}{t})^{(\frac{1-\eps}{36})t}$ samples. Here, $c$ and $c_1$ are large enough constants.
\end{theorem}
Note that when $k$ is a constant, this theorem gives a sample-memory tradeoff even for $\Omega(n)$ memory.

Our argument yields a slightly better quantitative lower bound for the special case when $P$ is the parity function, that is, $P(x)=(\sum_{i=1}^k x^i )\mod 2$. We will represent this function by $\pari$.

\begin{theorem}\label{thm:sublinearparity}
Let $0<\epsilon<1-3\frac{\log 24}{\log n}$ and $P$ be the parity predicate $\pari$ on $k=t$ bits. Suppose there's an algorithm that distinguishes between $D_{null}$ and $D_{planted}$ with probability at least $1/2 +s$ and uses $<n^{\epsilon}$ memory. Then for $k\le  n/c$\footnote{$c$ is a large enough constant}, if $s >  3 (\frac{n}{k})^{-(\frac{1-\eps}{18})k}$, the algorithm requires $ (\frac{n}{k})^{(\frac{1-\eps}{18})k}$ samples.
\end{theorem}

We prove both Theorem \ref{thm:pred} and \ref{thm:sublinearparity} (in Section \ref{section:proofssublinear}) via the same sequence of steps except for a certain quantitative bound presented in Lemma \ref{cl:spafourpred}. In the next subsection, we give an algorithm that takes $\widetilde{O}(n^{\eps}+k)k$ memory and $\widetilde{O}(n^{(1-\eps)k})$ samples, and distinguishes between $D_{null}$ and $D_{planted}$ for any predicate $P$, with probability $99/100$. Thus, the lower bounds are almost tight up to constant factors in the exponent for the parity predicate. The question of whether there exists an algorithm that runs in $O(n^{(1-\eps)t})$ samples and $O(n^\eps)$ memory, and distinguishes between $D_{null}$ and $D_{planted}$ with high probability, for $t$-resilient predicates $P$, remains open.

\subsection{Tightness of the Lower Bound}\label{section:tightness3}
In this section, we observe that there exists an algorithm $A$ that takes $O((n^{\eps}+k)\cdot k\log n)$ memory and $O(n^{(1-\eps)k}\cdot (n^{\eps}+k))$ samples, and distinguishes between $D_{null}$ and $D_{planted}$ for any predicate $P$, with probability $99/100$ (for $n^\eps> 10$).  

$A$ runs over $4n^{(1-\eps)k}\cdot (n^\eps+k)$ samples and stores the first $2(n^\eps +k)$ samples $(a,b)\in [n]^{(k)}\times\{0,1\}$ such that $a^i\le n^\eps+k, \forall i\in[k]$, that is, the bit $b$ depends only on the first $n^\eps+k$ bits of $x$ under the distribution $D_{planted}(x)$. If $A$ encounters less than $2(n^\eps +k)$ samples of the above mentioned form, $A$ outputs 1 (``planted"). Otherwise, $A$ goes over all the possibilities of the first $n^\eps+k$ bits of $x$ ($2^{n^\eps+k}$ possibilities in total) and checks if it could have generated the stored samples. If there exists a $y\in\{0,1\}^{n^\eps+k}$ that generated the stored samples, $A$ outputs 1 (``planted"), otherwise $A$ outputs 0. 

It's easy to see that $A$ uses $m=4n^{(1-\eps)k}\cdot (n^\eps+k)$ samples and at most $2(n^{\eps}+k)\cdot k\log n$ memory (as it takes only $k\log n$ memory to store a sample). Next, we calculate the probability of success. Let $Z_j$ be a random variable as follows: $Z_j=1$ if the $j^{th}$ sample $(a_j,b_j)$ is such that $a_j^i\le n^\eps+k, \forall i\in[k]$ and 0 otherwise. 
\[\Pr[Z_j=1]=\frac{|[n^\eps+k]^{(k)}|}{|[n]^{(k)}|}\ge n^{-(1-\eps)k}\]
And $\Ex[\sum_{j=1}^{m}Z_j]=4(n^\eps+k)$. By Chernoff bound, $\Pr[\sum_j Z_j< 2(n^\eps +k)]\le e^{-\frac{4(n^\eps+k)}{8}}\le\frac{1}{100}$. Therefore, the probability that $A$ stores $2(n^\eps +k)$ samples is at least $99/100$. It's easy to see that $A$ always outputs 1 when the samples are generated from $D_{planted}(x)$ for any $x$. 

The probability that $A$ outputs 1, given that it stored $2(n^\eps +k)$ samples, when the samples are generate from $D_{null}$ is equal to the probability that there exists a $y\in\{0,1\}^{n^\eps+k}$ that could have generated the stored samples. Let $(a'_1,b'_1),...,(a'_{2(n^\eps +k)},b'_{2(n^\eps +k)})$ be the stored samples. There are at most $2^{n^\eps +k}$ sequences of $b'_1,...,b'_{2(n^\eps +k)}$ generated by some $y$ given $\{a'_j\}_{j\in[2(n^\eps +k)]}$. As, under $D_{null}$, $b$ is chosen uniformly at random from $\{0,1\}$, the probability that there exists $y\in\{0,1\}^{n^\eps+k}$ that could have generated the stored samples is at most $\frac{2^{n^\eps +k}}{2^{2(n^\eps +k)}}=2^{-(n^\eps +k)}\le 1/100$. Hence, the probability of success is at least $99/100$.

\subsection{Proof of Theorems \ref{thm:pred} and \ref{thm:sublinearparity}}\label{section:proofssublinear}
Fix a $t$-resilient $k$-ary predicate $P$. Let $B$ be a branching program of width $d$ and length $m$ that has a success probability of $p = 1/2 +s$ for distinguishing between $D_{null}$ and $D_{planted}(x)$ ($x$ is uniformly distributed over $\zo^n$) for the predicate $P$. 

We first use hybrid argument to obtain that the branching program must have a non-trivial probability of distinguishing with a single sample. Towards this, define $H_j(x)$ to be the distribution over $m$ samples where the first $j$ samples are drawn from $D_{planted}(x)$ and the remaining $m-j$ samples are from $D_{null}$.

From Equation~\eqref{eq:succprexactpred} for $B$, we obtain: 
\begin{align*}
&\frac{1}{2}\left(\Pr_{u_1,...,u_m\sim D_{null}}[B(u_1,...,u_m)=0]+1-\Pr_{\substack{x\in_R\{0,1\}^n;\\u_1,...,u_m\sim D_{planted}(x)}}[B(u_1,...,u_m)=0]\right)\ge \frac{1}{2}+s\\
\implies& \Pr_{\substack{x\in_R\{0,1\}^n;\\u_1,...,u_m\sim D_{null}}}[B(u_1,...,u_m)=0]-\Pr_{\substack{x\in_R\{0,1\}^n;\\u_1,...,u_m\sim D_{planted}(x)}}[B(u_1,...,u_m)=0]\ge 2s
\end{align*}

The above expression can be written as a telescopic sum of the distinguishing probabilities over the $m+1$ hybrids, $H_{j}(x), j\in\{0,...,m\}$.

\begin{align*}
&\Pr_{\substack{x\in_R\{0,1\}^n;\\(u_1,...,u_m)\sim H_0(x)}}[B(u_1,...,u_m)=0]-\Pr_{\substack{x\in_R\{0,1\}^n;\\(u_1,...,u_m)\sim H_m(x)}}[B(u_1,...,u_m)=0]\\
=& \sum_{j=1}^m{\left[\Pr_{\substack{x\in_R\{0,1\}^n;\\(u_1,...,u_m)\sim H_{j-1}(x)}}[B(u_1,...,u_m)=0]-\Pr_{\substack{x\in_R\{0,1\}^n;\\(u_1,...,u_m)\sim H_j(x)}}[B(u_1,...,u_m)=0]\right]}
\end{align*}
Thus, there is a $j'\in \{1,...,m\}$ such that 
\begin{equation}\label{eq:hypred}
\Pr_{\substack{x\in_R\{0,1\}^n;\\(u_1,...,u_m)\sim H_{j'-1}(x)}}[B(u_1,...,u_m)=0]-\Pr_{\substack{x\in_R\{0,1\}^n;\\(u_1,...,u_m)\sim H_{j'}(x)}}[B(u_1,...,u_m)=0]\ge \frac{2s}{m}
\end{equation}

Next, we will show that for $0<\eps<1-3\frac{\log(24)}{\log n}$, $d\le 2^{n^\eps}$ and $0<t\le k<n^{(1-\eps)/6}/3, n/c$\footnote{$c$ is a large enough constant.}, $B$ distinguishes between the hybrids $H_{j'-1}(x)$ and $H_{j'}(x)$ with probability of at most $p_t \le c_1(\frac{n}{t})^{-(\frac{1-\eps}{18})t}$ (where $c_1$ is a large enough constant). Therefore, $\frac{2s}{m}\le c_1(\frac{n}{t})^{-(\frac{1-\eps}{18})t}$.

When $P = \pari$ ($t=k$), we will achieve better bounds. We show that for $0<k<n/c$,  $B$ distinguishes between the hybrids $H_{j'-1}(x)$ and $H_{j'}(x)$ with probability of at most $p'_k \le 5(\frac{n}{k})^{-(\frac{1-\eps}{9})k}$. Therefore, $\frac{2s}{m}\le 5(\frac{n}{k})^{-(\frac{1-\eps}{9})k}$.

Theorems \ref{thm:pred} and \ref{thm:sublinearparity} follows through following observations:

\begin{enumerate}
\item For $0<\eps<1-3\frac{\log(24)}{\log n}$ and $0<t\le k<\{n^{(1-\eps)/6}/3, n/c$\}, as $\frac{2s}{m}\le c_1(\frac{n}{t})^{-(\frac{1-\eps}{18})t}$, if $m\le (\frac{n}{t})^{(\frac{1-\eps}{36})t}$, then $s\le \frac{c_1}{2}\cdot (\frac{n}{t})^{-(\frac{1-\eps}{36})t}$.
\item When $P = \pari$, $t=k$, for $0<\eps<1-3\frac{\log(24)}{\log n}$ and $0< k< n/c$, as $\frac{2s}{m}\le 5(\frac{n}{k})^{-(\frac{1-\eps}{9})k}$, if $m\le (\frac{n}{k})^{(\frac{1-\eps}{18})k}$, then $s\le \frac{5}{2}\cdot (\frac{n}{k})^{-(\frac{1-\eps}{18})k}$.
\end{enumerate}

Now, we are ready to prove the upper bound on the capabilities of $B$ in distinguishing between the hybrids $H_{j'-1}(x)$ and $H_{j'}(x)$.

For $0<t\le k<n/c$, let $d_t=(\frac{n}{t})^t$. \\

Let $\L_i$ be the set of vertices in the layer-$i$ of the branching program $B$. Let $E_j(v)$ represent the event and $P_j(v)$ be the probability of reaching the vertex $v$ of the branching program $B$ when $x$ is picked uniformly at random from $\{0,1\}^n$ and the $m$ samples are drawn from $H_j(x)$. Let $E_j(0)$ represents the event of $B$ outputting $0$ when $x$ is picked uniformly at random from $\{0,1\}^n$ and the $m$ samples are drawn from $H_j(x)$. Let $v_1$ and $v_2$ be two vertices in the branching program such that $v_1$ is in an earlier layer than $v_2$. 
Let 
$P_j(v_2\mid v_1)$ be the probability of reaching the vertex $v_2$ of the branching program $B$ given that the computational path also reached the vertex $v_1$, when $x$ is picked uniformly at random from $\{0,1\}^n$ and the $m$ samples are drawn from $H_j(x)$.
Let 
$Q_j(v)$ be the probability of the branching program outputting 0 given that it reached vertex $v$, when $x$ is picked uniformly at random from $\{0,1\}^n$ and the $m$ samples are drawn from $H_j(x)$. Note that if $v$ is a vertex in layer-$i$ of the branching program such that $i\ge j$, $Q_j(v)$ does not change with the choice of $x$ as all the samples after the $j^{th}$ layer are independently drawn from $D_0$ ($D_{null}$).

Then, we can rewrite the expression on left hand side of Equation \ref{eq:hypred} as
\begin{align*}
\sum_{v_1\in \L_{j'-1},v_2\in\L_{j'}}\Pr[E_{j'-1}(0)\wedge E_{j'-1}(v_2)\wedge E_{j'-1}(v_1)]-\\
\sum_{v_1\in \L_{j'-1},v_2\in\L_{j'}}\Pr[E_{j'}(0)\wedge E_{j'}(v_2)\wedge E_{j'}(v_1)]
\end{align*}
 
For a vertex $v_2$ in the $j'^{th}$ layer, conditioned on the event $E_{j'}(v_2)$, event $E_{j'}(0)$ is independent of the event $E_{j'}(v_1)$. Similarly, conditioned on $E_{j'-1}(v_2)$, event $E_{j'-1}(0)$ is independent of the event $E_{j'-1}(v_1)$. And as the last $m-j'$ samples are drawn from the same distribution $D_0$, $Q_{j'}(v_2)=Q_{j'-1}(v_2)$.

For a vertex $v_1$ in the $(j'-1)^{th}$ layer, both $P_{j'}(v_1)$ and $P_{j'-1}(v_1)$ are equal to the probability of reaching the vertex $v_1$, when $x$ is picked uniformly at random from $\{0,1\}^n$ and the first $j'-1$ samples are drawn from $D_1(x)$. 

Hence, the expression can be rewritten as

\begin{align*}
&\sum_{v_1\in \L_{j'-1},v_2\in\L_{j'}}P_{j'-1}(v_1)\cdot P_{j'-1}(v_2\mid v_1)\cdot Q_{j'-1}(v_2)-\\
&~~~~~~~~~\sum_{v_1\in \L_{j'-1},v_2\in\L_{j'}}P_{j'}(v_1)\cdot P_{j'}(v_2\mid v_1)\cdot Q_{j'}(v_2)\\
=&\sum_{v_1\in \L_{j'-1},v_2\in\L_{j'}}P_{j'-1}(v_1)\cdot Q_{j'}(v_2)\cdot \left(P_{j'-1}(v_2\mid v_1)-P_{j'}(v_2\mid v_1)\right)\\
=&\sum_{v_2\in\L_{j'}}{Q_{j'}(v_2)\left(\sum_{v_1\in \L_{j'-1}}P_{j'-1}(v_1)\cdot \left(P_{j'-1}(v_2\mid v_1)-P_{j'}(v_2\mid v_1)\right)\right)}
\end{align*}

Let $L$ be the set of vertices in the layer-$(j'-1)$ of the branching program $B$ such that $\forall v_1\in L, \;P_{j'-1}(v_1)\ge d^{-1}d_t^{-1}$. Then, the above expression, can be rewritten as 

\begin{align}
\nonumber
=&\sum_{v_2\in\L_{j'}}{Q_{j'}(v_2)\left(\sum_{v_1\in L}P_{j'-1}(v_1)\cdot \left(P_{j'-1}(v_2\mid v_1)-P_{j'}(v_2\mid v_1)\right)\right)}\\
\nonumber
&+~~~~~~~~~\sum_{v_2\in\L_{j'}}{Q_{j'}(v_2)\left(\sum_{v_1\not \in L}P_{j'-1}(v_1)\cdot \left(P_{j'-1}(v_2\mid v_1)-P_{j'}(v_2\mid v_1)\right)\right)}\\
\nonumber
\le&\sum_{v_2\in\L_{j'}}{Q_{j'}(v_2)\left(\sum_{v_1\in L}P_{j'-1}(v_1)\cdot \left(P_{j'-1}(v_2\mid v_1)-P_{j'}(v_2\mid v_1)\right)\right)}\\
\nonumber
&+~~~~~~~~~\sum_{v_1\not \in L}P_{j'-1}(v_1)\left(\sum_{v_2\in\L_{j'}}{Q_{j'}(v_2)\cdot P_{j'-1}(v_2\mid v_1)}\right)\\
\label{eq:expred}
\le &\sum_{v_2\in\L_{j'}}{Q_{j'}(v_2)\left(\sum_{v_1\in L}P_{j'-1}(v_1)\cdot \left(P_{j'-1}(v_2\mid v_1)-P_{j'}(v_2\mid v_1)\right)\right)}+\frac{1}{d_t}
\end{align}

The last inequality follows from the fact that the width of the branching program is $d$, for a vertex $v_1$ not in $L$, $P_{j'-1}(v_1)$ is at most $\frac{1}{d\cdot d_t}$ and that the summation of the expression over $v_2$ can be at most $1$. 

Let $\P_{x|E_{j'}(v)}$ be the probability distribution of the random variable $x$ conditioned on the event $E_{j'}(v)$. For notational easiness, we will also denote this distribution by $\P_{x|v}$. We claim that for all $v_1\in L$, the distribution $\P_{x|v_1}$ has min-entropy of at least $(n-\log (d)-\log(d_t))$, that is, $\forall x'\in \{0,1\}^n, \; \P_{x|v_1}(x')\le d\cdot d_t\cdot 2^{-n}$.

The proof is as follows: as $x$ is chosen uniformly at random from $\{0,1\}^n$, for all $x'$, 
\[\sum_{v_1\in L}\Pr[x=x'\wedge E_{j'}(v_1)]\le \Pr[x=x']\le 2^{-n}\]

This implies,
\begin{align}
\nonumber
&\sum_{v_1\in L}\Pr[x=x'\mid E_{j'}(v_1)]\cdot P_{j'}(v_1)\le 2^{-n}\\
\nonumber
\implies & \Pr[x=x'\mid E_{j'}(v_1)]\le 2^{-n}\cdot \frac{1}{P_{j'}(v_1)}=2^{-n}\cdot \frac{1}{P_{j'-1}(v_1)}\le  2^{-n}\cdot \frac{1}{d^{-1}d_t^{-1}}\\
\label{eq:minentpred}
\implies &\P_{x|v_1}(x')\le d\cdot d_t\cdot 2^{-n}
\end{align}

Let $S_{(v_1,v_2)}$ be the set of all the labels $(a,b)\in [n]^{(k)}\times\{0,1\}$ such that the edge labeled $(a,b)$ at vertex $v_1$, goes into vertex $v_2$ in the next layer. Let $\P_{(a_{j'},b_{j'})|v_1}$ represent the distribution of the $j'$ sample conditioned on the event $E_{j'}(v_1)$, when $x$ is chosen uniformly at random from $\{0,1\}^n$ and the $m$ samples are chosen from $H_{j'}(x)$. As the $j'^{th}$ sample is drawn from the distribution $D_1(x)$, for every $a\in [n]^{(k)}$,
\begin{align}
\nonumber
\P_{(a_{j'},b_{j'})|v_1}(a,b)=&\sum_{x'\in\{0,1\}^n}\Pr[x=x'\mid E_{j'}(v_1)]\cdot \Pr[(a_{j'},b_{j'})=(a,b)\mid x=x']\\
\label{eq:samplepred}
=&\frac{1}{|\Sk|}\cdot \left(\sum_{x':P(x'^a)=b }\P_{x|v_1}(x')\right)
\end{align}
This is because $a$ is chosen uniformly from $\Sk$ and conditioned on $x$, the $j'$th sample is independent of $v_1$.

When the samples are drawn from $H_{j'-1}(x)$, the $j'^{th}$ sample is drawn from $D_0$ and is independent of the event $E_{j'-1}(v_1)$. Therefore, the probability of $j'^{th}$ sample being $(a,b)$ in this hybrid is $\frac{1}{2|\Sk|}$ for all $a\in \Sk, b\in\{0,1\}$. 

For every $v_1\in L$, we can rewrite the expression $\left(P_{j'-1}(v_2\mid v_1)-P_{j'}(v_2\mid v_1)\right)$ as follows:
\begin{align}
\nonumber
P_{j'-1}(v_2\mid v_1)-P_{j'}(v_2\mid v_1)&=\sum_{(a,b)\in S_{(v_1,v_2)}}\left(\P_{(a_{j'},b_{j'})|E_{j'-1}(v_1)}(a,b)-\P_{(a_{j'},b_{j'})|v_1}(a,b)\right)\\
\nonumber
&=\sum_{(a,b)\in S_{(v_1,v_2)}}\left(\frac{1}{2|\Sk|}-\P_{(a_{j'},b_{j'})|v_1}(a,b)\right)\\
\nonumber
&=\sum_{(a,b)\in S_{(v_1,v_2)}}\left(\frac{1}{2|\Sk|}-\frac{1}{|\Sk|}\cdot \left(\sum_{x': P(x'^a)=b }\P_{x|v_1}(x')\right)\right)\\
\label{eq:biaspred}
&=\frac{1}{|\Sk|}\cdot\left(\sum_{(a,b)\in S_{(v_1,v_2)}}\left(\frac{1}{2}-\sum_{x': P(x'^a)=b }\P_{x|v_1}(x')\right)\right)
\end{align}

Next, we will show that $\forall v_1\in L$, the above expression $\left|\frac{1}{2}-\sum_{x': P(x'^a)=b }\P_{x|v_1}(x')\right|$ is small for most samples $(a,b)$ (Lemma \ref{cl:spafourpred}). Define $T_l=\{\a\in\{0,1\}^n : \sum_{i=1}^n\a^i=l\}$ for $l\in \N$.\\

\begin{lemma}\label{cl:spafourpred}
For all $v_1\in L$, $0<\eps<1-3\frac{\log 24}{\log n}$, $0< t\le k<\{\frac{n}{c},n^{\frac{(1-\eps)}{6}}/3\}$, 
$$\left|\sum_{x'}\P_{x|v_1}(x')\cdot (-1)^{P(x'^a)}\right|\le c_2 n^{-(\frac{1-\eps}{18})t}$$
for all but at most $c_2 n^{-(\frac{1-\eps}{18})t}$ fraction of $a\in\Sk$ (recall that P is a $t$-resilient $k$-ary predicate).

For all $v_1\in L$, $0<\eps<1-3\frac{\log 24}{\log n}$, $0< k<\frac{n}{c}$, 
$$\left|\sum_{x'}\P_{x|v_1}(x')\cdot (-1)^{\pari(x'^a)}\right|\le 2\left(\frac{n}{k}\right)^{-(\frac{1-\eps}{9})k}$$
for all but at most $2(\frac{n}{k})^{-(\frac{1-\eps}{9})k}$ fraction of $a\in \Sk$.

 Here, $c$ and $c_2$ are large enough constants.
\end{lemma}

Before, we prove the Lemma, we prove the following claim:

\begin{claim}\label{cl:spafour}
For all $v_1\in L$, $0<\eps<1-3\frac{\log 24}{\log n}$, $0< t\le l<\frac{n}{c}$, 
$$\Ex_{\a\in T_l}\left(\sum_{x'}\P_{x|v_1}(x')\cdot (-1)^{\inner{\a,x'}}\right)^2\le 2 \left(\frac{n}{l}\right)^{-(\frac{1-\eps}{3})l}$$
\end{claim}
\begin{proof}
As $v_1\in L$, using Equation \ref{eq:minentpred}, we know that for all values of $x'$, $\P_{x|v_1}(x')\le 2^{n^\eps-n}\cdot d_t$. 
\begin{align*}
\Ex_{\a\in T_l}\left(\sum_{x'}\P_{x|v_1}(x')\cdot (-1)^{\inner{\a,x'}}\right)^2
=&\Ex_{\a\in T_l}{\left(\sum_{x',x''}\P_{x|v_1}(x')\cdot \P_{x|v_1}(x'')\cdot (-1)^{\inner{\a,x'+x''}}\right)}\\
=&\sum_{x',x''}\P_{x|v_1}(x')\cdot \P_{x|v_1}(x'')\cdot \Ex_{\a\in T_l} {(-1)^{\inner{\a,x'+x''}}}\\
=&\sum_{x',z}\P_{x|v_1}(x')\cdot \P_{x|v_1}(z+x')\cdot \Ex_{\a\in T_l} {(-1)^{\inner{\a,z}}}
\end{align*}
Let $\B_{T_l}(\delta)=\{\gamma\in\{0,1\}^n\mid  |\Ex_{\a\in T_l} {(-1)^{\inner{\a,\gamma}}}|>\delta\}$. Using Lemma \ref{lem:spafour} \cite{KRT17}, we know that for $1\ge \delta\ge (\frac{8l}{n})^{\frac{l}{2}}$,
\[|\B_{T_l}(\delta)|\le 2e^{-\delta^{2/l}\cdot n/8}\cdot 2^n\]

We can rewrite the expression as follows:
\begin{align*}
&\Ex_{\a\in T_l}\left(\sum_{x'}\P_{x|v_1}(x')\cdot (-1)^{\inner{\a,x'}}\right)^2
=\sum_{z\in \B_{T_l}(\delta)}\sum_{x'}\P_{x|v_1}(x')\cdot \P_{x|v_1}(z+x')\cdot \Ex_{\a\in T_l} {(-1)^{\inner{\a,z}}}\\
&~~~~~~~~~~~~~~~~~~~~~~~~~~~~~~~~~~~~~~~~~~~~~~~~~~~~~~~~~~~~~~~~~~
+\sum_{z\not\in \B_{T_l}(\delta)}\sum_{x'}\P_{x|v_1}(x')\cdot \P_{x|v_1}(z+x')\cdot \Ex_{\a\in T_l} {(-1)^{\inner{\a,z}}}\\
\le &\sum_{z\in \B_{T_l}(\delta)}\sum_{x'}\P_{x|v_1}(x')\cdot \P_{x|v_1}(z+x')\cdot |\Ex_{\a\in T_l} {(-1)^{\inner{\a,z}}}|\\
&~~~~~~~~~~~~~~~~~~~~~~~~~~~~~~~~~~~~~~~~~~~~~~~~~~~~~~~~~~~~~~~~~~
+\sum_{z\not\in \B_{T_l}(\delta)}\sum_{x'}\P_{x|v_1}(x')\cdot \P_{x|v_1}(z+x')\cdot |\Ex_{\a\in T_l} {(-1)^{\inner{\a,z}}}|\\
\le &\sum_{z\in \B_{T_l}(\delta)}\sum_{x'}\P_{x|v_1}(x')\cdot \P_{x|v_1}(z+x')
+\sum_{z\not\in \B_{T_l}(\delta)}\sum_{x'}\P_{x|v_1}(x')\cdot \P_{x|v_1}(z+x')\cdot \delta\\
\le&\sum_{z\in \B_{T_l}(\delta)}\sum_{x'}\P_{x|v_1}(x')\cdot \P_{x|v_1}(z+x')+\delta\\
\le&~~~~2^{n^\eps-n}\cdot d_t\cdot |\B_{T_l}(\delta)|+\delta\\
\le&~~~~2e^{-\delta^{2/l}\cdot n/8}\cdot 2^{n^\eps}\cdot d_t+\delta\\
\end{align*}
The second inequality follows from the definition of $\B_{T_l}(\delta)$ and the fact that $ |\Ex_{\a\in T_l} {(-1)^{\inner{\a,z}}}|\le 1$ always. The fourth inequality follows from Equation \ref{eq:minentpred} as $ \P_{x|v_1}(z+x')\le 2^{n^\eps-n}\cdot d_t, \forall x',z$. And, the last inequality follows from the bound on $|\B_{T_l}(\delta)|$.

Therefore, $\forall \delta, (\frac{8l}{n})^{\frac{l}{2}}\le \delta\le 1$, $\Ex_{\a\in T_l}\left(\sum_{x'}\P_{x|v_1}(x')\cdot (-1)^{\inner{\a,x'}}\right)^2\le 2e^{-\delta^{2/l}\cdot n/8}\cdot 2^{n^\eps}\cdot d_t+\delta$.\\

For $0<t\le l<n/c$, $d_t= (\frac{n}{t})^t\le (\frac{n}{l})^l$. Let $\delta =(\frac{n}{l})^{-(\frac{1-\eps}{3})l}$. As $l<n/c$ where $c$ is large enough constant, $(\frac{8l}{n})^{\frac{l}{2}}\le \delta$. Therefore, 
\begin{align*}
\Ex_{\a\in T_l}\left(\sum_{x'}\P_{x|v_1}(x')\cdot (-1)^{\inner{\a,x'}}\right)^2
&\le 2e^{-(\frac{n}{l})^{-(1-\eps)\frac{2}{3}}\cdot n/8}\cdot 2^{n^\eps}\cdot \left(\frac{n}{l}\right)^l+\left(\frac{n}{l}\right)^{-(\frac{1-\eps}{3})l}\\
&=  2e^{-(\frac{n}{l})^{\frac{1}{3}(1-\eps)+\eps}\cdot l/8}\cdot2^{n^\eps}\cdot \left(\frac{n}{l}\right)^l+\left(\frac{n}{l}\right)^{-(\frac{1-\eps}{3})l}\\
&\le 2^{-(\frac{n}{l})^{\frac{1}{3}(1-\eps)+\eps}\cdot l/8+n^\eps+l\log(n/l)+1}+\left(\frac{n}{l}\right)^{-(\frac{1-\eps}{3})l}\\
&\le 2^{-l\log(n/l)}+\left(\frac{n}{l}\right)^{-(\frac{1-\eps}{3})l}
\le 2 \left(\frac{n}{l}\right)^{-(\frac{1-\eps}{3})l}
\end{align*}
Here, the inequalities follow by assuming that $l< n/c$, for large enough $c$ such that $(n/l)^{1/3}> 36\log (n/l)$, and $0<\eps<1-3\frac{\log 24}{\log n}$ such that $n^{\frac{1}{3}(1-\eps)}\ge 24$. The second last inequality follows from the following calculations.
\begin{align*}
\left(\frac{n}{l}\right)^{\frac{1}{3}(1-\eps)+\eps}\cdot l/8= \frac{1}{24}n^{\frac{1}{3}(1-\eps)}n^\eps l^{\frac{2}{3}(1-\eps)} +\frac{3}{36} l\left(\frac{n}{l}\right)^{\frac{1}{3}} \left(\frac{n}{l}\right)^{\frac{2}{3}\eps}> n^\eps +3l\log(n/l)
\end{align*}

This proves the claim.
\end{proof}

\begin{proof} \textbf{of Lemma \ref{cl:spafourpred}}.
We first prove the statement for the general $t$-resilient $k$-ary predicate $P$ for $0<\eps<1-3\frac{\log 24}{\log n}$ such that $0<t\le k<\{\frac{n}{c},n^{\frac{(1-\eps)}{6}}/3\}$. Using Claim \ref{cl:spafour}, we know that for all $0<t\le l<\frac{n}{c}$, 
$$\Ex_{\a\in T_l}\left(\sum_{x'}\P_{x|v_1}(x')\cdot (-1)^{\inner{\a,x'}}\right)^2\le 2  \left(\frac{n}{l}\right)^{-(\frac{1-\eps}{3})l}$$
We consider the following expression, 
$$\Ex_{a\in \Sk}\left(\sum_{x'}\P_{x|v_1}(x')\cdot (-1)^{P(x'^a)}\right)^2$$

Substituting $(-1)^P$ with its Fourier expansion, we get 
\begin{align}
\nonumber
\Ex_{a\in \Sk}\left(\sum_{x'}\P_{x|v_1}(x')\cdot (-1)^{P(x'^a)}\right)^2
&=\Ex_{a\in \Sk}\left(\sum_{x'}\P_{x|v_1}(x')\cdot \sum_{\alpha\subseteq [k]} \hat{P}(\alpha)\chi_\alpha(x'^a)\right)^2\\
\label{eq:53}
&=\Ex_{a\in \Sk}\left(\sum_{\alpha\subseteq [k], |\alpha|\ge t}  \hat{P}(\alpha)\sum_{x'}\P_{x|v_1}(x')\cdot\chi_\alpha(x'^a)\right)^2\\
\label{eq:51}
&\le  \Ex_{a\in \Sk}\left(\sum_{\alpha\subseteq [k], |\alpha|\ge t} \hat{P}(\alpha)^2\right)\left(\sum_{\alpha\subseteq [k],|\alpha|\ge t} \left(\sum_{x'}\P_{x|v_1}(x')\cdot\chi_\alpha(x'^a)\right)^2 \right)\\
\label{eq:52}
&=  \Ex_{a\in \Sk}\sum_{\alpha\subseteq [k],|\alpha|\ge t} \left(\sum_{x'}\P_{x|v_1}(x')\cdot\chi_\alpha(x'^a)\right)^2 \\
\nonumber
&=  \sum_{j=t}^{k}\left(\sum_{\alpha\subseteq [k],|\alpha|=j}\Ex_{a\in \Sk} \left(\sum_{x'}\P_{x|v_1}(x')\cdot\chi_\alpha(x'^a)\right)^2\right) 
\end{align}

Equality \ref{eq:53} follows from the fact that $P$ is $t-$resilient. Inequality \ref{eq:51} follows from Cauchy-Schwarz. Equality \ref{eq:52} follows from Parseval's identity. Let $v(a,\alpha)$ be a $n$-bit vector defined as follows: $\forall i\in[k]$, set $v(a,\alpha)^{a^i}=1$ if only if $\alpha_i=1$ (and 0 otherwise). It's easy to see that $\chi_\alpha(x'^a)=(-1)^{\inner{x',v(a,\alpha)}}$ and $|v(a,\alpha)|=|\alpha|$. And, for a fixed $\alpha$, when $a$ is uniformly distributed over $\Sk$, $v(a,\alpha)$ is uniformly distributed over $T_{|\alpha|}$. Therefore, the above expression can be rewritten as,
\begin{align}
\nonumber
\Ex_{a\in \Sk}\left(\sum_{x'}\P_{x|v_1}(x')\cdot (-1)^{P(x'^a)}\right)^2&=  \sum_{j=t}^{k}\left(\sum_{\alpha\subseteq [k],|\alpha|=j}\Ex_{v\in T_j} \left(\sum_{x'}\P_{x|v_1}(x')\cdot(-1)^{\inner{v,x'}}\right)^2\right) \\
\label{eq:55}
&\le   2\cdot \sum_{j=t}^{k}\left(\sum_{\alpha\subseteq [k],|\alpha|=j} \left(\frac{n}{j}\right)^{-(\frac{1-\eps}{3})j} \right)\\
\nonumber
&\le2\cdot \sum_{j=t}^{k}\left({k \choose j}\cdot  \left(\frac{n}{j}\right)^{-(\frac{1-\eps}{3})j} \right)\\
\label{eq:56n}
&\le c'\cdot \sum_{j=t}^{k}\left(\left(e\frac{k}{j}\right)^j\cdot  \frac{j^{(\frac{1-\eps}{3})j}}{(n^{\frac{1-\eps}{3}})^j} \right)\\
\label{eq:57n}
&\le c'\cdot \sum_{j=t}^{k} n^{-\frac{(1-\eps)}{6}j} \le 2c' \cdot  (n^{-\frac{(1-\eps)}{6}t})
\end{align}
Inequality \ref{eq:55} follows from Claim \ref{cl:spafour}. For large enough $c'$, Inequality \ref{eq:56n} follows from Sterling's bound on factorials. As $k\le n^{\frac{1-\eps}{6}}/3$ and assuming $n^{\frac{1-\eps}{6}}\ge 2$, Inequality \ref{eq:57n} follows. Therefore,
$$\Ex_{a\in \Sk}\left(\sum_{x'}\P_{x|v_1}(x')\cdot (-1)^{P(x'^a)}\right)^2\le c'' \cdot n^{-\frac{(1-\eps)}{6}t}$$ 
for large enough constant $c''$.
From this expression, it's easy to see that 
$$\left|\sum_{x'}\P_{x|v_1}(x')\cdot (-1)^{P(x'^a)}\right|\le c_2 \cdot n^{-\frac{(1-\eps)}{18}t}$$
for all but at most $c_2 \cdot n^{-\frac{(1-\eps)}{18}t}$ fraction of $a\in\Sk$ ($c_2=c''^{1/3}$).\\

Now, we prove the lemma for the special case of $P=\pari$. As parity function is symmetric, 
\begin{align*}
\Ex_{a\in \Sk}\left(\sum_{x'}\P_{x|v_1}(x')\cdot (-1)^{\pari(x'^a)}\right)^2=\Ex_{\a\in T_k}\left(\sum_{x'}\P_{x|v_1}(x')\cdot (-1)^{\inner{\a,x'}}\right)^2
\end{align*}
Therefore, using Claim \ref{cl:spafour}, we have shown that 
$$\Ex_{a\in \Sk}\left(\sum_{x'}\P_{x|v_1}(x')\cdot (-1)^{\pari(x'^a)}\right)^2\le \delta_k=2 \left(\frac{n}{k}\right)^{-(\frac{1-\eps}{3})k}$$
for $0<k<n/c$.

From these expressions, it's easy to see that $|\sum_{x'}\P_{x|v_1}(x')\cdot (-1)^{\pari(x'^a)}|\ge (\delta_k)^{1/3}$ for at most $(\delta_k)^{1/3}$ fraction of the values of $a\in \Sk$ which completes the proof of the lemma. 

\end{proof}

Now, we come back to bounding the capabilities of $B$ to distinguishing between the $j'-1$ and $j'$ hybrids ($p_t$, $p'_t$). Substituting the expression for
$P_{j'-1}(v_2\mid v_1)-P_{j'}(v_2\mid v_1)$ obtained from Equation \ref{eq:biaspred} in Equation \ref{eq:expred}, we get that 
\begin{align}
\nonumber
&\Pr_{\substack{x\in_R\{0,1\}^n;\\(u_1,...,u_m)\sim H_{j'-1}(x)}}[B(u_1,...,u_m)=0]-\Pr_{\substack{x\in_R\{0,1\}^n;\\(u_1,...,u_m)\sim H_{j'}(x)}}[B(u_1,...,u_m)=0]\\
\nonumber
\le&  \sum_{v_2\in\L_{j'}}{Q_{j'}(v_2)\left(\sum_{v_1\in L}P_{j'-1}(v_1)\cdot \left(\frac{1}{|\Sk|}\cdot\left(\sum_{(a,b)\in S_{(v_1,v_2)}}\left(\frac{1}{2}-\sum_{x': P(x'^a)=b }\P_{x|v_1}(x')\right)\right)\right)\right)}+\frac{1}{d_t}\\
\nonumber
=&  \sum_{v_1\in L}{P_{j'-1}(v_1)\left(\sum_{v_2\in \L_{j'}}Q_{j'}(v_2)\cdot \left(\frac{1}{|\Sk|}\cdot\left(\sum_{(a,b)\in S_{(v_1,v_2)}}\left(\frac{1}{2}-\sum_{x': P(x'^a)=b }\P_{x|v_1}(x')\right)\right)\right)\right)}+\frac{1}{d_t}\\
\label{eq:1pred}
\le &\sum_{v_1\in L}{P_{j'-1}(v_1)\left(\sum_{v_2\in \L_{j'}}Q_{j'}(v_2)\cdot \left(\frac{1}{|\Sk|}\cdot\left(\sum_{(a,b)\in S_{(v_1,v_2)}}\left|\frac{1}{2}-\sum_{x': P(x'^a)=b }\P_{x|v_1}(x')\right|\right)\right)\right)}+\frac{1}{d_t}\\
\label{eq:2pred}
\le &\sum_{v_1\in L}{P_{j'-1}(v_1)\left(\frac{1}{|\Sk|}\cdot\sum_{v_2\in \L_{j'}}\sum_{(a,b)\in S_{(v_1,v_2)}}\left|\frac{1}{2}-\sum_{x': P(x'^a)=b }\P_{x|v_1}(x')\right|\right)}+\frac{1}{d_t}\\
\label{eq:3pred}
= &\sum_{v_1\in L}{P_{j'-1}(v_1)\left(\Ex_{a\in_R \Sk}\sum_{b\in\{0,1\}}\left|\frac{1}{2}-\sum_{x': P(x'^a)=b}\P_{x|v_1}(x')\right|\right)}+\frac{1}{d_t}\end{align}
Inequality \ref{eq:1pred} follows just from taking the absolute values. Inequality \ref{eq:2pred} follows from the fact that $Q_{j'}(v_2)\le 1$ for all $v_2\in \L_{j'}$. Equality \ref{eq:3pred} follows from the fact that each edge labelled by $(a,b)\in \Sk\times\{0,1\}$ goes into some vertex in the next layer $\L_{j'}$.

For the general $t-$ resilient $k$-ary predicate $P$, Lemma \ref{cl:spafourpred} showed that, for all but
 $c_2 n^{-(\frac{1-\eps}{18})t}$ fraction of $a\in\Sk$, 
$$\sum_{b\in\{0,1\}}\left|\frac{1}{2}-\sum_{x': P(x'^a)=b}\P_{x|v_1}(x')\right|\le c_2\cdot n^{-(\frac{1-\eps}{18})t}$$
As the maximum value of this expression is 1, we have shown that 
$$p_t\le \sum_{v_1\in L}{P_{j'-1}(v_1)\cdot 2c_2\cdot n^{-(\frac{1-\eps}{18})t}+\left(\frac{n}{t}\right)^{-t}\le c_1\left(\frac{n}{t}\right)^{-(\frac{1-\eps}{18})t}}$$
for large enough constant $c_1=2c_2 +1$.\\

For the special case of $P=\pari$, using a similar argument, we can show that for $0<k<n/c$, as $d_k=\left(\frac{n}{k}\right)^k$,

$$p'_k\le \sum_{v_1\in L}{P_{j'-1}(v_1)\cdot 2\cdot 2 \left(\frac{n}{k}\right)^{-(\frac{1-\eps}{9})k}+\left(\frac{n}{k}\right)^{-k}\le5\cdot \left(\frac{n}{k}\right)^{-(\frac{1-\eps}{9})k}}.$$

This completes the proofs of Theorem \ref{thm:pred} and Theorem \ref{thm:sublinearparity}.

\clearpage
\section{Acknowledgements}
We would like to thank Avishay Tal and David Woodruff for the discussions about the problem of distinguishing subspaces.
\bibliographystyle{alpha}
\bibliography{biblo,dblp,scholar,mathreview,custom,crypto}

\end{document}